\documentclass{article}

\usepackage{booktabs} 
\usepackage[ruled]{algorithm2e} 
\usepackage{natbib}
\usepackage{authblk}
\usepackage{fullpage}

\SetAlFnt{\small}
\SetAlCapFnt{\small}
\SetAlCapNameFnt{\small}
\SetAlCapHSkip{0pt}
\IncMargin{-\parindent}
\usepackage{amsmath, amssymb, amsthm}
\usepackage{tikz}
\usetikzlibrary{shapes,arrows,fit,calc,positioning}
\usetikzlibrary{arrows}
\usetikzlibrary{decorations.markings}
\usepackage{enumerate}
\usepackage{hyperref}
\usepackage{graphicx}
\usepackage{caption}
\usepackage{hyperref}
\usepackage{subcaption}
\usepackage{mathtools}
\usepackage{comment}
\usepackage{algorithmic}
\usepackage[mathscr]{eucal}

\theoremstyle{definition}
\newtheorem{theorem}{Theorem}

\newtheorem{lemma}[theorem]{Lemma}

\newtheorem{definition}[theorem]{Definition}

\theoremstyle{remark}
\newtheorem{claim}[theorem]{Claim}
\newtheorem{remark}[theorem]{Remark}

\newcommand{\opt}{\mathsf{opt}}

\newcommand{\pls}{N}						 
\newcommand{\plu}{\cal N}				 	
\newcommand{\gr}{G}							  
\newcommand{\gru}{\cal G}				    
\newcommand{\vrts}{V}						  
\newcommand{\edgs}{E}						 
\newcommand{\s}{\mathbf{S}}	
\newcommand{\str}{S}			   
\newcommand{\games}{\mathbf{\Gamma}}		  
\newcommand{\game}{\Gamma}			
\newcommand{\prot}{\Xi}						   
\def \load {\ell}										  
\def \total {d}										  

\def \opt {\mathrm{OPT}}
\def \o {\mathbf{\o}}

\def \NWA {Never-Walk-Alone }

\newcounter{note}[section]

\begin{document}
	
	\title{Resource-Aware Protocols for Network Cost-Sharing Games}

	\author[1]{George Christodoulou}
	\author[2]{Vasilis Gkatzelis}
	\author[3]{Mohamad Latifian}
	\author[1]{Alkmini Sgouritsa}
	\affil[1]{University of Liverpool, \texttt{\{gchristo,alkmini\}@liverpool.ac.uk}}
	\affil[2]{Drexel University, \texttt{gkatz@drexel.edu}}
	\affil[3]{Sharif University of Technology, \texttt{latifian@ce.sharif.edu}}

	\maketitle
	
	\begin{abstract}
  We study the extent to which decentralized cost-sharing protocols
  can achieve good price of anarchy (PoA) bounds in network
  cost-sharing games with $n$ agents. We focus on the model of
  resource-aware protocols, where the designer has prior access to the
  network structure and can also increase the total cost of an edge
  (overcharging), and we study classes of games with concave or convex
  cost functions. We first consider concave cost functions and our
  main result is a cost-sharing protocol for symmetric games on
  directed acyclic graphs that achieves a PoA of $2+\varepsilon$ for some arbitrary 
	small positive $\varepsilon$, which improves to $1+\varepsilon$ for games with at 
	least two players. We also achieve a PoA of 1 for series-parallel 
	graphs and show that no protocol can achieve a PoA better than $\Omega(\sqrt{n})$ for multicast
  games. We then also consider convex cost functions and prove
  analogous results for series-parallel networks and multicast games,
  as well as a lower bound of $\Omega(n)$ for the PoA on directed
  acyclic graphs without the use of overcharging.
\end{abstract}

	\section{Introduction}
In this paper we study the classic model of network cost-sharing games introduced by 
\citet{ADKTWR08} and \citet{CRV10}, where a set of $n$ agents need to use a path that connects 
a designated source to a sink in a given graph $G$. Each edge of the graph
is characterized by a cost function, and the cost of each edge needs to be distributed
among the agents that use it in their chosen paths. \citet{ADKTWR08} considered
cost functions that are constant, i.e., independent of the number of agents that use 
the edge, and assumed that the cost of each edge is shared \emph{equally}
among its users	(which is known as the \emph{Shapley cost-sharing protocol}). 
This gave rise to a, now, very well-studied class of games where each agent strategically 
chooses a path that minimizes her share of the total cost, anticipating the strategic 
choices of the other agents. It is well-known that the Nash equilibria of this game can 
have a high social cost, and a lot of subsequent work has focused on evaluating exactly 
how much greater this cost can be compared to the optimal social cost.

Deviating from this line of work which assumed that the cost-sharing
protocol for each edge is predetermined to be the Shapley protocol,
\citet{CRV10} instead approached the problem from the perspective of a
designer and asked a compelling question: \emph{can we design a better
  cost-sharing protocol?} In particular, if the designer of a network
wishes to minimize the social cost in equilibrium, what cost-sharing
protocol should she choose? Focusing on graphs with constant cost
functions, they proposed a list of desirable properties that the
cost-sharing protocols should satisfy, they provided characterizations
of these protocols, and proved tight bounds regarding their
performance. Following the same research agenda, 
\citet{vFH13} provided tight bounds for more general cost
functions, but for more restricted network structures, namely for
scheduling games.  Our work is motivated by the same question and the
main differences are two-fold: \emph{i)}~we consider more general
classes of cost functions than \cite{CRV10} and more general network
structures than \cite{vFH13}, and \emph{ii)} we
focus on \emph{resource-aware} protocols which were recently
introduced by \citet{CS16} and studied by \citet{CGS17}.

The class of resource-aware protocols provide a refinement regarding the amount of 
\emph{information} available to the protocol and they provide a 
 middle-ground between the two extreme information models that were studied in the past. Note that the
amount of information that is available to the protocol plays a central role. For instance, 
can the decisions of the protocol regarding how to distribute the cost of 
an edge depend on the structure of the graph $G$ beyond that edge? Can they depend
on the set of agents that are using the network at any given time? The more
information the protocol has access to regarding the state of the system, the 
greater the ability of the designer to reach efficient outcomes. 
Most of the prior work focused on either \emph{omniscient} protocols that have 
{\em full knowledge} of the instance (the structure of the graph and its cost functions
as well the set of agents that are using the system), or \emph{oblivious} protocols that
have {\em no knowledge} of the instance, except the set of users who use the corresponding
edge. Previous work has also referred to these protocols as ``non-uniform'' and ``uniform'',
respectively.

These assumptions lie at two extremes of the information spectrum.
The former applies to a very static, or centralized, system where
each edge always has access to up-to-date information regarding 
the users and the graph $G$. The latter, on the other hand, is very
pessimistic, assuming that the cost-sharing decisions of each edge
need to be oblivious to the state of the system. 
The resource-aware model strikes a balance between these two extremes
by assuming that the protocol of each edge knows graph $G$ and the edge 
cost functions, but is aware only of the set of users who chose to utilize 
that particular edge.
This model is more realistic than the oblivious one when the 
network does not change often, which is the case in 
many real systems. In that case, the structure of the graph can be used 
to inform the design of the cost-sharing protocol of each edge and this 
protocol needs only to be updated when the structure of the network is
changed.
Of course, the fact that a resource-aware protocol has additional
information compared to an oblivious one is only interesting if this
can lead to improved outcomes. Understanding the ways in which this
information can be useful is one of the central goals of this paper.

In designing cost-sharing protocols, the goal is to ensure that the worst-case 
equilibria of the induced game approximate the optimal social cost within some 
small factor. We measure the performance of these protocols using the worst-case price of anarchy 
(PoA) measure, i.e., the ratio of the social cost in the worst equilibrium 
over that in the optimal solution. Essentially, once the designer selects the protocol 
on each edge, then an adversary chooses the requested subset of agents 
so that the PoA of the induced game is maximized.

\paragraph{Overcharging.}
In accordance with some of the recent work on resource-aware protocols~\cite{CGS17},
we also consider protocols that use overcharging. This is in contrast to some of the 
prior work (e.g., \cite{CRV10},\cite{vFH13} and \cite{CS16}) which was restricted to
{\em budget-balanced} protocols, i.e., protocols such that the costs
of the users using some edge add up to {\em exactly} the cost of
the edge. Our protocols can instead choose to charge
the users additional costs in order to dissuade them from using some of the edges.
It is important to note that, once we introduce increased costs in our
protocols, we compare the performance of the equilibria in the induced
game with the {\em increased} costs to the original optimal solution
with the {\em initial} cost functions. Therefore, using this technique
requires a careful trade-off between the benefit of improved incentives
and the drawback of penalties suffered.

\subsection{Our Results}
In this paper we provide several upper and lower bounds for the price of anarchy (PoA)
of interesting classes of network cost-sharing games. We consider \emph{symmetric} 
games, where the source and sink of every agent is the same, as well as \emph{multicast} 
games, where the sink is the same for every agent, but the source may be different.
In terms of the graphs that we consider, we go beyond parallel link networks and
consider \emph{directed acyclic graphs} as well as \emph{series-parallel graphs}.
Finally, we consider different types of functions that describe the costs of the 
network edges, and we group our results based on whether these functions
are \emph{concave} or \emph{convex}.

It is for \textbf{concave cost functions} that we get the main result of the paper
(see Section~\ref{sec:concaveDAG}). That is, a resource-aware cost-sharing protocol, 
the \NWA protocol, which guarantees a PoA of $2+\varepsilon$ for symmetric games on 
arbitrary directed acyclic graphs and some arbitrarily small constant $\varepsilon>0$. 
In fact, for instances that involve more than one agent the PoA of this 
protocol drops to $1+\varepsilon$ for an arbitrarily small constant $\varepsilon$. We show that every 
game that this protocol gives rise to is guaranteed to possess a pure Nash equilibrium, 
and  \emph{every} such equilibrium corresponds to an optimal assignment of agents to paths. 
The success of this protocol in such a large class of instances crucially depends on 
the fact that the cost-sharing decisions depend on the structure of the graph, and we 
leverage this power in a novel way, using careful overcharging in order to avoid suboptimal 
equilibria despite the 
decentralized nature of the protocol. Apart from this result, we also provide an alternative 
protocol that achieves a PoA of 1 for all symmetric instances with series-parallel graphs 
and strictly concave costs, without using any overcharging; with a minor amount of overcharging, 
this protocol can also be adapted to get a PoA of $1+\varepsilon$ even for weakly concave costs (Section~\ref{sec:concave-SP}). 
We complement these positive results with two impossibility results for the class of (non-symmetric) 
multicast games, for which we show that no resource-aware protocol can achieve a PoA better than
$\Omega(n)$ without overcharging and $\Omega(\sqrt{n})$ even with overcharging (Section~\ref{sec:multicast}).

We then transition to \textbf{convex cost functions} in Section~\ref{sec:convex}. 
We show that for symmetric games a simple protocol can achieve a PoA of 1 for 
series-parallel graphs, but for directed acyclic graphs no budget-balanced protocol 
can achieve a PoA better than $\Omega(n)$. Also, for multicast games, we show that
no protocol can achieve a PoA better than $\Omega(\sqrt{n})$, even if it uses overcharging.

\section{Related Work}

The papers that are most closely related to ours are that of \citet{CS16} 
and \citet{CGS17}, which were the first ones to study the design of
resource-aware protocols for network cost-sharing games. \citet{CS16}
considered constant cost functions, as in \citet{CRV10}, and they
showed that for outerplanar graphs there exists a resource-aware
protocol that performs better than any oblivious one. On the negative
side, they showed that for general graphs the best protocol performs
asymptotically the same with the best oblivious one. \citet{CGS17}
considered more general classes of cost functions beyond constant
ones, as in \citet{vFH13}, but were restricted to parallel link networks. 
In this paper we
combine the strengths of these two papers by studying general classes 
of cost functions for larger classes of networks beyond parallel links.

Apart from the aforementioned papers, there are are several others that
also focus on the design and analysis of cost-sharing protocols. \citet{HF14}
applied cost-sharing protocols to capacitated facility location games and,
among other results, showed that omniscient protocols can achieve efficiency
for convex cost functions and general strategy spaces.
\citet{MW13} studied multiple cost-sharing protocols in a model of utility 
maximization instead of cost-minimization. \citet{GMW14} provide a complete 
characterization of the space of oblivious cost-sharing protocols that always 
guarantee the existence of a pure Nash equilibrium and showed that it corresponds 
to the set of generalized weighted Shapley value protocols. Leveraging this
characterization, \citet{GKR16} analyzed this family of cost-sharing protocols 
and showed that the unweighted Shapley value achieves the optimal price of anarchy 
guarantees for network cost-sharing games. Also, \citet{HM11} 
studied the performance of several cost-sharing protocols in a
slightly modified setting, where each player declares a different
demand for each resource. 
\citet{CLS16} studied network design games with constant cost
functions under the Bayesian setting, where the position of the
players' sources on the graph is drawn from a distribution over all
vertices. They considered overcharging, where they could use any
non-budget-balanced policy under the restriction of preserving
budget-balance in all equilibria. 
 More recently~\citet{HHHS18} studied cost-sharing
in a model that imposes some constraints over the portions of the cost that
can be shared among the agents.

There is also a very close connection between network cost-sharing games and the 
very well-studied class of congestion games  (e.g., \cite{R73b}, \cite{M96}, \cite{MS96b}, 
\cite{AAE05}, \cite{GS07}, \cite{BGR10}, \cite{HK10}, \cite{HKM11}). In particular, a congestion game
can be interpreted as the game that arises from a network cost-sharing setting
if the costs of the resources are divided equally among the agents. For instance, 
a congestion game with cost functions that are polynomials of degree $d$ can
be cast as a network cost-sharing game with costs that are polynomials of degree
$d+1$ combined with the equal sharing protocol. 

The impact of cost-sharing methods on the quality of equilibria has
also been studied in other models: \citet{MS01} focused on
participation games, while \citet{M08} and \citet{MR09} studied queueing
games. Also, very closely related in spirit is previous work on
coordination mechanisms, beginning with the work of \citet{CKN09} and subsequently
in the papers of \cite{I+09, AJM08, C09, AH12, K13, C+13, CMP14, BIKM14}. Most work on
coordination mechanisms concerns scheduling games and how the price of
anarchy varies with the choice of local machine policies (i.e., the
order in which to process jobs assigned to the same machine). Some connections
between cost-sharing policies and coordination mechanisms are also provided in~\cite{CGV17}.

\section{Preliminaries}
We study the performance of cost-sharing protocols on classes of network cost-sharing games.
A \emph{class} of network cost-sharing games, $\games$, is defined by a tuple $(\plu, \gru, \mathcal C, \prot)$, 
which comprises a universe of players $\plu$, a universe of graphs $\gru$, a universe of cost functions $\cal C$, 
and a cost-sharing protocol $\prot$ (formally defined later on). An \emph{instance} of a cost-sharing game $\game \in \games$ then consists of a set of players $\pls \subseteq \plu$, a graph $\gr = (\vrts, \edgs) \in \gru$ where each edge $e \in \edgs$ is assigned a cost function $c_e$ drawn from $\cal C$, and a cost-sharing protocol $\prot$. Each player $i \in \pls$ is associated with a source $s_i$ and a sink $t_i$ in the graph, and she needs to use a path in $\gr$ that connects them. For each edge $e$, the cost function $c_e(\load)$ is a non-decreasing function that satisfies $c_e(0)=0$, and indicates the cost of this edge when the \emph{load} on this edge, i.e., the number of players using it, is $\load$. 

The outcome of a game is a strategy profile $\s = (\str_1, \str_2, ... , \str_{n})$, where 
$\str_i$ is the path that player $i$ chose to connect her source and sink. For each edge $e \in \edgs$, let $\load_e(\s)$ be the number of players using this edge in their path, i.e., $\load_e(\s) = |\{i\in \pls: e \in \str_i\}|$. The cost of $e$ in this strategy profile is $c_e(\load_e(\s))$, and this cost needs to be covered by the players using this edge.
In this paper we design \emph{cost-sharing protocols}, i.e., 
protocols that decide how the cost of each edge will be distributed among 
its users. A cost-sharing protocol~$\prot$ defines at each strategy profile $\s$ a 
cost share $\xi_{ie}(\s)$ for each edge $e\in E$ and every agent $i\in \pls$ that uses this edge in $\s$, i.e., such that
$e\in \str_i$. For agents $i$ with $e \notin \str_i$ we have $\xi_{ie}(\s)=0$, so only the agents using an edge are responsible
for its costs. We denote the total cost-share of player $i$ in $\s$ as
$$\xi_i(\s) = \sum_{e \in E} \xi_{ie}(\s).$$

{\bf Pure Nash Equilibrium (PNE).} The goal of every user is to minimize her
cost-share, so different cost-sharing protocols lead to different classes 
of games and, hence, to possibly very different outcomes. The efficiency 
of a game, thus, crucially depends on the choice of the protocol. In evaluating
the performance of a cost-sharing protocol, we measure the quality of the
pure Nash equilibria that arise in the games that it induces. A 
strategy vector~$\s$ is a \emph{pure Nash equilibrium} (PNE) of a game $\game$ 
if for every player $i \in \pls$, and every other path $\str'_i$ with the same source and sink as $\str_i$, we have
\begin{equation*}
\xi_i(\s)~=~\xi_i(\str_i, \s_{-i}) ~~\le~~ \xi_i(\str'_i, \s_{-i}),
\end{equation*}
where $\s_{-i}$ denotes the vector of strategies for all players other
than $i$. In other words, in a PNE no player can decrease her cost share by 
unilaterally deviating from path $\str_i$ to $\str'_i$ if all the other players' 
choices remain fixed. In accordance with prior work, we restrict our attention 
to protocols that possess at least one PNE for every $\game \in \games$, which are 
called \emph{stable} protocols.

{\bf Price of Anarchy (PoA).} 
To evaluate the efficiency of a strategy profile $\s$, we use the total cost
$C(\s)=\sum_{e\in \edgs}c_e(\load_e(\s))$, and we quantify the performance
of the cost-sharing protocol using the price of anarchy measure. 
Given a cost-sharing protocol $\Xi$, the \emph{price of anarchy} (PoA)
of the induced class of games $\games=(\plu, \gru, \cal C, \Xi)$  
is the worst-case ratio of equilibrium cost to optimal cost over all 
games in $\games$. That is, if $Eq(\game)$ is the set of pure Nash equilibria 
of the game $\game$, and $F(\game)$ the set of all strategy profiles of $\game$, then
\[\text{PoA}(\games)=\sup_{\game \in \games} \frac{\max_{\s\in Eq(\game)}C(\s)}{\min_{\s^*\in F(\game)}C(\s^*)}.\]

{\bf Budget-balance and Overcharging.} We say that a cost-sharing protocol is \emph{budget-balanced} if for every
edge $e$ and profile $\s$ we have $\sum_{i\in \pls}\xi_{ie}(\s)=c_e(\load_e(\s))$, i.e., the cost shares that the 
protocol distributes to the players using an edge adds up to exactly the cost of that edge.
Apart from budget-balanced mechanisms, we also consider mechanisms that may use overcharging.
These mechanisms define for each edge $e$ a cost function $\hat{c}$ such that $\hat{c}_e(\load)> c_e(\load)$ for 
some values of $\load$. As a result, the social cost of a given strategy profile $\s$ may be increased from $C(\s)$ to 
$\hat{C}(\s)=\sum_{e\in \edgs}\hat{c}_e(\load_e(\s))$. In these mechanisms, we measure the quality of the equilibria
using the new costs, but we compare their performance to the optimal solutions
based on the original cost functions:
\[\text{PoA}(\games)=\sup_{\game \in \games} \frac{\max_{\s\in Eq(\game)}\hat{C}(\s)}{\min_{\s^*\in F(\game)}C(\s^*)}.\]

{\bf Classes of Games.} We evaluate the price of anarchy for several classes of games that may differ
in the strategies of the players, the structure of the graph, or the cost functions used for the edges. 
We call a cost-sharing game \emph{symmetric} if all the players have the same source $s$ and sink $t$, 
and hence the same set of strategies (all the paths from $s$ to $t$). Most of this paper focuses 
on symmetric cost-sharing games, but we also consider games where players have different sources and the 
same sink, known as \emph{multicast} cost-sharing games. In terms of classes of graphs, the two main ones
that we consider are \emph{directed acyclic graphs} (DAGs) and \emph{series-parallel graphs}. The former
include all directed graphs that do not have a cycle, while the latter are defined recursively using
two simple composition operations (see Section~\ref{sec:series-parallel-prelims} for a formal definition).
When it comes to classes of cost functions, the two main ones that we consider are convex functions (exhibiting 
non-decreasing marginal costs), or concave cost functions (exhibiting non-increasing marginal costs). 

{\bf Informational Assumptions.} Throughout this paper we focus on the design of 
\emph{resource-aware} cost-sharing protocols. The decisions of these protocols regarding 
how to share the cost of an edge $e$ can depend on the set of players using the edge, the 
structure of the network, and the cost functions of its machines, but it cannot depend 
on the set of agents that are participating in the game, but are not using edge $e$. This
captures the intuition that these cost-sharing mechanisms are decentralized and are not
aware of the set of players that is using the network at any given time. On the other hand,
since the structure of the network and the characteristics of its edges are not expected to
change radically, these protocols can be designed in a way that leverages this information.
This is in contrast to the class of \emph{oblivious} protocols that are unaware of anything
other than the set of players that use the edge or the class of \emph{omniscient} protocols
that are assumed to have access to all the information regarding the instance. In this paper 
we design resource-aware protocols that take advantage of the additional
information that they have access to (relative to oblivious ones) in order to achieve 
good price of anarchy bounds in general classes of instances.

{\bf Global Ordering.} Some of our mechanisms, as well as many protocols in the related work 
(e.g., \cite{Mou99,CGS17}), decide how to distribute the cost among the agents for a class
of games $\games =(\plu, \gru, \cal C, \Xi)$ by using a global ordering over the universe 
$\plu$ of players. Then, for any instance $\game \in \games$ from this class with a set of
players $\pls \subseteq \plu$, this global ordering also implies an ordering of the agents
in $\pls$ that are actually participating. Based on a global ordering, we use $h_e(\s)$ to
denote the highest priority agent using edge $e$ in strategy profile $\s$, and we refer to
her as the \emph{leader} on this edge.

\section{Directed Acyclic Graphs with Concave Cost Functions}\label{sec:concaveDAG}
In this section we consider a directed acyclic graph $G$ with two designated 
vertices $s$ and $t$, and concave cost functions on the edges. We
study symmetric $n$-player games where each player strives to establish 
a connection from $s$ to $t$. In Section~\ref{sec:leader-based} we define 
a family of protocols, which we call {\em leader-based protocols}, and in 
Section \ref{sec:NWA-protocol} we present our main result which is a
leader-based protocol that uses overcharging and achieves a PoA of
almost $2$.

We also provide some evidence that overcharging is needed in order to
achieve constant PoA, if one focuses on a subset of leader-based
protocols that we call {\em static-share leader-based protocols}; we
show (in Section~\ref{sec:StaticShare}) that, even for strictly concave
cost functions, no such budget-balanced protocol can achieve a PoA better 
than ${\tilde{\Omega}}(\sqrt{n})$. This is in contrast to the case of 
parallel links where prior work has shown that such a protocol can achieve 
a PoA of $1$ for strictly concave cost functions~\cite{CGS17}. We leave as 
an open question the existence of a budget-balanced leader-based (or any 
resource-aware) protocol with constant PoA.

\subsection{Structure of the Optimal Network}
We first observe a property of the optimal solution, i.e., the assignment of users to paths 
that minimizes the total cost, for symmetric games. Specifically, we show that when the cost
functions are concave there always exists an optimal solution where all the users are assigned
to the same path connecting the source to the sink.

\begin{lemma}
	\label{lem:optPath}
	Let $G$ be a graph with a designated source $s$ and sink $t$, and with concave cost functions on the edges.
	In any symmetric instance where all the users need to connect from $s$ to $t$, there exists an optimal solution 
	which uses a single path from $s$ to $t$ for all the users.
\end{lemma}

Lemma~\ref{lem:optPath} implies that for any graph $G$ with source $s$ and sink $t$, and any number of agents $n$, there exists at least one path
from $s$ to $t$ such that assigning all $n$ players to that path minimizes the total cost. In fact, note that we can easily find this optimal path
for any number of agents, $n$, by computing the shortest path from $s$ to $t$ assuming the cost of each edge $e$ is equal to $c_e(n)$.
The following definition provides us with a way to consistently choose one of these paths, which will be useful in defining our protocols.
\begin{definition}($\opt(\load)$).
	\label{def:concaveOptPath}
	For a given symmetric instance, let $\opt(\load)$ be an optimal path from $s$ to $t$ when the total load, i.e., the
	total number of agents in the system, is $\load$. When there are multiple such paths, $\opt(\load)$ breaks ties arbitrarily 
	but consistently.  
\end{definition}

\subsection{Leader-based Protocols}
\label{sec:leader-based}
In this work we consider a class of protocols that we call {\em leader-based protocols}. The idea behind these protocols is that for
each edge we (consistently) identify a single player to be the {\em leader} and charge that player some load-dependent amount $\psi_e(\load)$,
while the rest of the players share the remaining cost equally. We study resource-aware protocols, so the choice of the $\psi_e(\load)$
values may depend on the structure of the graph and the cost functions of all the edges.

\begin{definition}\label{def:leader-based}(Leader-based Protocol).
Given a class of games $\games$, a {\em leader-based protocol} uses a predetermined priority ordering $\pi$ over the universe of players $\plu$. For any strategy profile $\s$, the protocol identifies a {\em leader} on each edge as the highest priority player (according to $\pi$) using edge $e$, denoted as $h_e(\s)$. The protocol further defines a value $\psi_e(\load)\leq c_e(\load)$ for each edge $e$ and load $\load$. Then, for each strategy profile $\s$, the cost-share of player $i$ for using edge $e$ is 
	
	\begin{equation*}
	\xi_{ie}(\s) =
	\begin{cases}
	\psi_e(\load_e(\s)) & \text{if } i=h_e(\s)\\
	\frac{c_e(\load_e(\s))-\psi_e(\load_e(\s))}{\load_e(\s)-1} & \text{otherwise}.
	\end{cases} 
	\end{equation*}
\end{definition}

We remark that the protocol used in \cite{CGS17} to achieve a PoA of 1 for parallel-link graphs lies in the class of leader-based protocols, with $\psi_e(\load_e(\s))$ being a fixed value when $e \in \opt(\load_e(\s))$ and $\psi_e(\load_e(\s))=c_e(\load_e(\s))$ when $e \notin \opt(\load_e(\s))$. We call these protocols {\em static-share leader-based protocols} and show in Section~\ref{sec:StaticShare} that their PoA is $\Omega\left(\frac{\sqrt{n}}{\log^2n}\right)$ for DAGs. Note that $\opt(\load_e(\s))$ corresponds to the path that the optimal solution would use if the total load in the system was $\load_e(\s)$, and a resource-aware protocol has access to this information since it depends only on the graph and the edge costs; not on the \emph{actual} load in the system, which may be different than $\load_e(\s)$.

\subsection{An Almost Efficient Leader-based Protocol with Overcharging}\label{sec:NWA-protocol}
In this section we design a protocol that uses overcharging and has PoA $=2+\varepsilon$ for an arbitrarily small constant $\varepsilon>0$; in fact, for instances with $n>1$ users the PoA becomes $1+\varepsilon$. This protocol lies in the family of leader-based protocols but it is not a static-share leader-based protocol. The main idea behind this protocol is that the cost share that the leader pays for each edge decreases as the cost shares of the other players increase. The cost share of the leader is infinitesimal when he is not the only one using an edge and therefore he has an incentive to never be alone at any edge; for this reason, we call this protocol the {\em \NWA} protocol. 
  
Our protocol uses a property of DAGs, according to which we can assign weights on the edges such that for \emph{any} two paths with the same endpoints the sum of their edges' weights are equal. Those weights are used in order to guarantee that the charges of the leader are positive, and the equality between the weights of alternative paths is used in our main theorem in order to compare the charges on those paths.

\begin{lemma}\label{lem:weights}
	Consider a directed acyclic graph $G=(V, E)$. We can assign to each edge $e\in E$ an integer weight $w_e > 0$, and to each vertex $v\in V$ a value $\omega_v \geq 0$, such that for any two vertices $u$ and $v$ that are connected through a path $p$ in $G$, we have 
	$\sum_{e \in p} w_e  = \omega_v-\omega_u.$
\end{lemma}
\begin{proof}
Since $G$ is a DAG, we can topologically sort the vertices in $V$ so that all the edges are directed from left to right. We assign to each vertex $v$ a value $\omega_v$ equal to its position in the sorted ordering, and to each edge $e=(u,v)$ we assign the integer weight $w_e=\omega_v-\omega_u$. Since the vertices are topologically sorted, for every edge $(u,v)$ we have $\omega_u < \omega_v$, which implies that $w_e>0$ for every edge $e$. In other words, the weight of each edge corresponds to the number of vertices ahead in the topological sorting it points to. It is then easy to verify that for any two vertices, the sum of the weights of the edges for any path that connects them will be exactly their ``distance'' in this topological ordering.
\end{proof}

\textbf{\NWA Protocol.} Let $\opt(\total)$ be an optimal path from $s$ to $t$ when the total load is $\total$ (according to Definition~\ref{def:concaveOptPath}), let $C>2\sum_e c_e(|\plu|)$ be an arbitrarily large number greater than the total cost that may ever appear, and let $\epsilon>0$ be an arbitrarily small value satisfying $\epsilon < \min_{e'}\{c_{e'}(1)\}/\sum_e w_e$, where $w_e$ is the weight of each edge $e$ implied by Lemma~\ref{lem:weights}. The \NWA protocol defines a new cost function $\hat{c}_e$ for every edge $e$, satisfying $\hat{c}_e(\load)\geq c_e(\load)$ for any load $\load$, as follows:
\[
\hat{c}_e(\load) =
\begin{cases}
2c_e(1) & \text{if } \load = 1\\
2 (\load-1)c_e(\load) +\epsilon_e(2c_e(\load)) & \text{if } e \notin \opt(\load) \text{ and } \load \neq 1\\
c_e(\load) + \epsilon_e\left(\frac{c_e(\load)}{\load -1}\right) & \text{otherwise.}
\end{cases}\,,
\qquad\quad \text{where}\quad \epsilon_e(x) = \frac{w_eC - x}{C}\epsilon   \,,
\]

Then, in order to decide how to share these new (overcharged) costs, the protocol considers a global ordering $\pi$ over the universe of players, and for any strategy profile $\s$ it identifies a {\em leader} on each edge, as the highest priority player $h_e(\s)$ using edge $e$. If there are at least two players using an edge, then the cost of the leader is the second term of the costs defined in cases 2 and 3 above, and the rest of the players equally share the remaining cost, i.e., the first term. In other words, for each strategy profile $\s$, each player $i$ using edge $e$ is charged:
\[
\xi_{ie}(\s) =
\begin{cases}
\zeta_{e}(\load_e(\s)) & \text{if } i \neq h_e(\s) \text{ or } \load_e(\s) = 1\\
\epsilon_e(\zeta_{e}(\load_e(\s))) & \text{otherwise,}
\end{cases}\,,
\]

\[
\text{where}\quad\zeta_{e}(\load) =
\begin{cases}
2 c_e(\load) & \text{if } e \notin \opt(\load) \text{ or } \load = 1\\
\frac{c_e(\load)}{\load-1} & \text{otherwise}
\end{cases}\,.
\]

Note that $C$ is used to guarantee that the leader does not pay a negative cost and $\epsilon$ is used to keep that overcharging arbitrarily small\footnote{We remark that for any load $\load>1$ and any $e\in \opt(\load)$ we can avoid even this small overcharging of $\epsilon\left(\frac{c_e(\load)}{\load -1}\right)$ by charging all players but the leader $\zeta_e(\load)=\frac{c_e(\load)-w_e\epsilon}{\load -1 -\epsilon/C}$ instead of $\frac{c_e(\load)}{\load-1}$ and the leader is still charged $\epsilon_e(\zeta_{e}(\load_e(\s)))$. For the sake of simplicity we keep this small overcharging.} and to guarantee that the highest priority player has an incentive not to be alone in any edge, unless the total number of players in the system is $1$, as stated in the following lemma.

\begin{lemma}
	\label{alone}
	For every game with concave cost functions and at least two players, when we use the \NWA protocol there is no Nash equilibrium $\s$ where the player $h$ with the highest priority according to $\pi$, among all active agents, is the only user of some edge.  
\end{lemma}

\begin{proof}
If player $h$ chooses a path where he is alone in some edge, then he pays at least $\min_{e'}\{c_{e'}(1)\}$. On the other hand, if this player chooses a path where he shares every edge with some agent, he is charged at most $\sum_e \epsilon_e(0) \geq \sum_e \epsilon_e(x)$, for any $x\geq 0$. Given the way we restricted $\epsilon$ to satisfy $\epsilon < \min_{e'}\{c_{e'}(1)\}/\sum_e w_e$, we have $\min_{e'}\{c_{e'}(1)\}>m \epsilon_e(0)$, so $h$ always prefers to share. Therefore, since he always has the option of sharing, e.g., by just following the path of some other player, a profile $\s$ where he is not sharing will never be a Nash equilibrium.
\end{proof}

A crucial subtlety in order to prove our main theorem is to guarantee
that no ties appear in the charges of alternative paths. We assume
that the mechanism first applies some arbitrarily small overcharging
using the following lemma (whose proof can be found in the appendix)
and then uses these cost functions in the protocol definition.

\begin{lemma}
	\label{lem:BreakTies}
	We can always increase the cost functions by adding arbitrarily small constants in a way that for any two vertices $u,v$, any two paths $p_1$ and $p_2$ from $u$ to $v$, and any strategy profiles $\s_1$, $\s_2$, we have 
	\[\sum_{e\in p_1}\zeta_{e}(\load_e(\s_1)) \neq \sum_{e\in p_2}\zeta_{e}(\load_e(\s_2)).\]
\end{lemma}

\begin{lemma}
	\label{lem:overStability}
	The \NWA protocol is stable, i.e., it always induces games that possess a pure Nash equilibrium.
\end{lemma}

\begin{proof}
	We show that for any number of agents, $n$, the optimal strategy profile $\s^*$ in which all of the load $\load=n$ is assigned to the optimal path $\opt(\load)$ is a pure Nash equilibrium. Consider any other path $p$ and let $p'$ and $\opt'$ be any two edge-disjoint sub-paths of $p$ and $\opt(\load)$, respectively, with the same endpoints $u$ and $v$, as in Figure~\ref{fig:2paths}. 
		
	\begin{figure}[h]
		\centering
		\tikzset{middlearrow/.style={
				decoration={markings,
					mark= at position 0.5 with {\arrow{#1}} ,
				},
				postaction={decorate}
			}
		}
		\begin{tikzpicture}[xscale=0.7,yscale=0.7]

		\node[draw,circle,thick,fill=black,inner sep=1.4pt,radius=0.25pt,label={[label distance=-1pt]below:{$s$}}] at (1,0) (s) {} ;
		\node[draw,circle,thick,fill=black,inner sep=1.4pt,radius=0.25pt,label={[label distance=-1pt]below:{$u$}}] at (4,0) (u) {} ;
		\node[draw,circle,thick,fill=black,inner sep=1.4pt,radius=0.25pt,label={[label distance=-1pt]below:{$v$}}] at (10,0) (v) {} ;
		\node[draw,circle,thick,fill=black,inner sep=1.4pt,radius=0.25pt,label={[label distance=-1pt]below:{$t$}}] at (13,0) (t) {} ;
		
		\draw[very thick, middlearrow={stealth}] (s) to (u);
		\draw[very thick, middlearrow={stealth}] (v) to (t);
		
		\draw[thick, middlearrow={stealth}] (u) to [bend left] node [above] {{$\opt'$}} (v);
		\draw[thick, middlearrow={stealth}] (u) to [bend right] node [above] {{$p'$}} (v);
		
		\end{tikzpicture}
		\caption{Network structure if players use one path.}
		\label{fig:2paths}
	\end{figure}
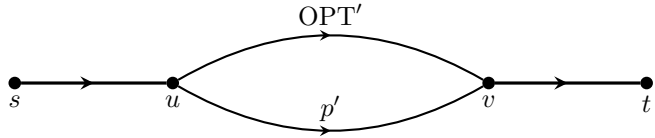
	
	Suppose that $\load >1$; the other case is trivial. If all the players use the $\opt(\load)$ path, then the cost share of each of them (including the leader) for the $\opt'(\load)$ portion of this path is at most
	$$\sum_{e\in \opt'} \frac{c_e(\load)}{\load-1} \leq \sum_{e\in p'} \frac{c_e(\load)}{\load-1} \leq \sum_{e\in p'} \frac{c_e(1)\load}{\load-1}\leq \sum_{e\in p'} 2c_e(1),$$
	where the last term is the cost share that any player should pay if they unilaterally deviate to $p'$. The first inequality is due to the optimality of $\opt(\load)$ and the second due to the concavity of the cost functions. Since this is true for any two edge-disjoint sub-paths of $p$ and $\opt(\load)$, we conclude that no player has an incentive to deviate to $p$.
\end{proof}

\begin{lemma}
	\label{lem:overOnlyOptPath}
	For any number of players $n$ and any non-optimal strategy profile $\s$ in which all the players use the same path $p$ from $s$ to $t$, $\s$ is not a Nash equilibrium 
	of the \NWA protocol.
\end{lemma}
\begin{proof}
	Similar to the previous proof, let $p'$ and $\opt'$ be two edge-disjoint sub-paths of $p$ and $\opt(\load)$, respectively, with the same endpoints (Figure~\ref{fig:2paths}); if there are many such paths, pick the one with total cost for load $\load=n$ strictly greater than that of $\opt'$ (there should be at least one, otherwise $p$ would be an optimal path). Based on the protocol, there is at least one player paying $2c_e(\load)$ for every edge $e$ in $p'$. Note that  
	$$ \sum_{e\in p'}2c_e(\load) < \sum_{e\in \opt'}2c_e(\load) \leq \sum_{e\in \opt'}2c_e(1)\, , $$
	which is the cost share if a player deviates to $\opt'$. Hence, that player has an incentive to deviate from $p'$ to $\opt'$.		
\end{proof}

\begin{lemma}
	\label{lem:overNoTwoPaths}
	For any strategy profile $\s$ in which players use at least two different paths, $\s$ is not a Nash equilibrium
	of the \NWA protocol.
\end{lemma}

\begin{proof}
Aiming for a contradiction, assume that there exists some instance where such a strategy profile $\s$ is a Nash equilibrium, and let $a_1$ denote the highest priority player in this instance. Let $p_1$ be the path that the highest priority player is using in $\s$, and let $(v^+, u)$ be the first edge along $p_1$ such that some player, denoted as $a^+$, uses edge $(v^+, u)$ without having previously used \emph{all} the previous edge in the $p_1$ path; in other words, $a_1$ and $a^+$ use paths that differ by at least one edge in getting from $s$ to $v^+$ but they both use edge $(v^+, u)$. If no such edge exists, then let $v^+=t$. Let $p'_1$ be the sub-path of $p_1$ from $s$ to $v^+$ used by $a_1$, and $p^+$ be the sub-path from $s$ to $v^+$ used by $a^+$. Since no player enters $p'_1$ before $v^+$ and due to Lemma \ref{alone}, which states that $a_1$ is not alone in any edge under $s$, there must exist some other player $a_2 \neq a_1$ that uses all of the edges in $p'_1$ (Figure \ref{twopaths}).   
	
	\begin{figure}[h]
		\centerline{\includegraphics[width=0.4\textwidth]{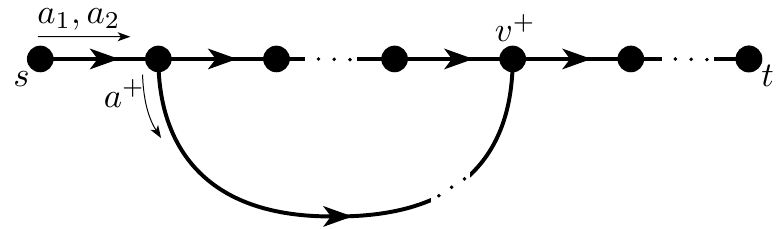}}
		\caption{Network structure if players use at least two paths.}
		\label{twopaths}
	\end{figure}
	
	Since $a_2$ is not the leader in any of the edges of $p'_1$ (because $a_1$ is), he is charged $\sum_{e \in p'_1} \zeta_{e}(\load_e(\s))$ for $p'_1$.  If he deviated to follow $p^+$ instead of $p'_1$ he would be charged at most $\sum_{e \in p^+} \zeta_{e}(\load_e(\s)+1)$. Since $\s$ is a Nash equilibrium, $a_2$ cannot improve his cost by deviating this way, so
	\begin{equation}
	\label{eq1}
	\sum_{e \in p'_1} \zeta_{e}(\load_e(\s))  < \sum_{e \in p^+} \zeta_{e}(\load_e(\s)+1)\,,
	\end{equation}
	where we have a strict inequality due to Lemma~\ref{lem:BreakTies}. However, since $a_1$ is the highest priority player among all players, his cost for $p'_1$ under $\s$ is
	$$\frac{(\omega_{v^+}-\omega_s)C-\sum_{e \in p'_1} \zeta_{e}(\load_e(\s))}{C} \epsilon >
	\frac{(\omega_{v^+}-\omega_s)C-\sum_{e \in p^+} \zeta_{e}(\load_e(\s)+1)}{C} \epsilon\,,$$
	where the inequality is due to \eqref{eq1} and the fact that $\sum_{e\in p'_1}w_e = \sum_{e\in p^+}w_e  = \omega_{v^+}-\omega_s$. But, the right hand side of this inequality would correspond to the cost of $a_1$ if he deviates from $p'_1$ to $p^+$.  Therefore $a_1$ has a reason to deviate from $p'_1$ to $p^+$ which contradicts the assumption that $\s$ is a Nash equilibrium.
\end{proof}

\begin{theorem}
	The PoA of the \NWA protocol for directed acyclic graphs with concave functions is at most $2+\varepsilon$, for an arbitrarily small value $\varepsilon>0$. If we assume $n>1$, then PoA $=1+\varepsilon$. 
\end{theorem}
\begin{proof}
	By Lemmas~\ref{lem:overStability},~\ref{lem:overOnlyOptPath} and \ref{lem:overNoTwoPaths} the only Nash equilibria are optimal paths so the PoA increases only due to overcharging.
	
		Recall that in order to avoid ties on players' charges for alternative paths we increase the cost functions by arbitrarily small amounts; let $\varepsilon_1$ be the total such increment. Moreover, let $\varepsilon_2=\sum_e \epsilon_e(0)$ be an upper bound on the share that the highest priority player may ever pay in any Nash equilibrium. Note that $\varepsilon_2$ depends on the choice of $\epsilon$ and therefore can be arbitrarily small. 

For $n=1$ the overcharging trivially results in PoA $\leq 2+\varepsilon_1$. For $n>1$, the total charges of all but the highest priority player sum up to the total cost occurred (after the increments for breaking ties on the cost shares) and therefore, PoA $\leq 1+ \varepsilon_1+\varepsilon_2$. 
\end{proof}

\subsection{Static-share Leader-based Protocols} 
\label{sec:StaticShare}
In this section we argue that the approach of \cite{CGS17}, which resulted in a
budget-balanced protocol for parallel links with PoA $=1$, cannot be successfully
applied here. The protocol used in \cite{CGS17} is a leader-based protocol where,
each edge essentially assumes that the set of agents currently using it are the 
only active users in the system. If, given that assumption, the particular edge
should be used in the optimal solution, then the leader is charged some carefully
chosen fixed cost, otherwise, the leader is charged the whole cost of the edge,
incentivizing that agent to deviate to another edge. This is an instance of a class
of protocols which we refer to as {\em static-share leader-based protocols}.

\begin{definition}(Static-share Leader-based Protocol).
	A {\em static-share leader-based protocol} is a leader-based protocol with the following choice of $\psi_e\leq c_e(\load)$: 
		\begin{equation*}
	\psi_e(\load) =
	\begin{cases}
	\psi_e & \text{if } e \in \opt(\load) \\
	c_e(\load) & \text{otherwise}
	\end{cases} 
	\end{equation*}
\end{definition}

We next show that {\em no} static-share leader-based protocol can guarantee a constant PoA. In particular, the price of anarchy
grows with the number of agents $n=|\pls|$ of the instance at hand. Our proof considers concave cost functions for simplicity, but we remark that the lower bound also holds for strictly concave cost functions as we describe in the proof. Due to limited space the proof of this theorem can be found in the appendix.

\begin{theorem}
	\label{thm:LB_SSLBprotocols}
	Any static-share leader-based protocol has PoA $= \Omega\left(\frac{\sqrt{n}}{\log^2n}\right)$ for directed acyclic graphs, even for the class of strictly concave cost functions.
\end{theorem}

\section{Series-Parallel Graphs with Concave Cost Functions}\label{sec:concave-SP}
In this section we study instances with series-parallel graphs and strictly concave cost 
functions and we are able to design a budget-balanced static-share leader-based protocol 
with PoA $=1$, which generalizes the protocol used in~\cite{CGS17} for parallel-link graphs. 
In fact, the same protocol extends to general concave cost functions results with PoA 
$=1+\epsilon$, for an arbitrarily small value $\epsilon>0$, after applying minor overcharging 
on the cost functions in order to transform them to strictly concave functions, similar to 
\cite{CGS17}.

\subsection{Preliminaries in Series-Parallel Graphs}\label{sec:series-parallel-prelims}
Here, we provide the basic definitions and properties of
directed series-parallel graphs (SPGs), based on~\cite{Epp92}.

A series-parallel graph (SPG) can be constructed by performing {\em
  series} and {\em parallel} compositions of smaller SPGs, starting
with copies of the basic SPG which is an edge. Each SPG has two
designated vertices/terminals called source and sink; regarding the
edge which is the basic SPG, one vertex serves as the source and the
other as the sink. One can construct any SPG through a (not unique)
sequence of the following operations:

\begin{itemize}
	\item Create a new graph, consisting of a single edge directed from $s$ to $t$. 
	\item Given two SPGs $C_1$ and $C_2$ with terminals $s_1,t_1,s_2,t_2$, form a new SPG $C$
				by merging $s_1$ and $s_2$ into one new terminal, $s$, and merging $t_1$ and $t_2$ into
				a new terminal, $t$. This is known as the {\em parallel composition} of $C_1$ and $C_2$.
	\item Given two SPGs $C_1$ and $C_2$, with terminals $s_1,t_1,s_2,t_2$, form a new SPG $C$, with
          terminals $s$, $t$, by merging $t_1$ and $s_2$ and identifying $s = s_1$ and $t = t_2$ as
					the new terminals. This is known as the {\em series composition} of $C_1$ and $C_2$.
\end{itemize} 

For the rest of the section we will consider a SPG $G=(V,M)$, with
source $s$ and sink $t$, and $C_1, \ldots , C_r$ is the sequence of
SPGs that are constructed in the process of generating $G$ (with
$C_r=G$). Let $\mathcal{C} = \{C_1,\ldots , C_r\}$ and $s_C, t_C$ be
the terminals of any component $C\in \mathcal{C}$.

\subsection{An Efficient Leader-based Protocol}
Here, we study symmetric games where each edge of $G$ has a {\em strictly} concave cost function, and
we present a budget-balanced cost-sharing protocol with PoA $=1$. In fact, as the following remark argues,
this also implies that we can achieve a PoA of $1+\epsilon$ for {\em weakly} concave cost functions for
an arbitrarily small constant $\epsilon>0$ if we use a very small amount of overcharging.

\begin{remark}
	Using similar overcharging arguments like the one used in~\cite{CGS17}, one can transform any 
	concave cost function into a \emph{strictly} concave one by replacing any linear portion of
	the function with a concave portion that is never more than $\epsilon$ times greater. This will 
	achieve a PoA equal to $1+\epsilon$, where $\epsilon>0$ is an arbitrarily small constant.
\end{remark}

\textbf{Protocol Definition.} Like all leader-based protocols, this protocol uses a global order $\pi$ to identify the 
leader $h_e(\s)$ of each edge for each profile $\s$. After carefully assigning values $\psi_e\leq c_e(\load_e(\s))$ on the edges of $G$, the 
cost share of player $i$ for using edge $e$ is defined as
\[
\xi_{ie}(\s) = \left\{
\begin{array}{l l}
  c_e(\load_e(\s)) & \quad \text{if}~~ e \notin \opt(\load_e(\s))~~\text{and}~~ i=h_e(\s)\\
  0 & \quad   \text{if}~~ e\notin \opt(\load_e(\s))~~\text{and}~~ i\neq h_e(\s)\\
  \psi_e & \quad  \text{if}~~ e\in \opt(\load_e(\s))~~\text{and}~~ i=h_e(\s)\\
  \frac{c_e(\load_e(\s)) - \psi_e}{\load_e(\s) - 1} & \quad  \text{if}~~e\in \opt(\load_e(\s))~~\text{and}~~ i\neq h_e(\s)
\end{array} \right.
\]

In order to complete the definition of the cost-sharing protocol we
need to define the $\psi_e$'s. In fact we define a value $\psi_C$ for
each component $C\in\mathcal{C}= \{C_1,\ldots , C_r\}$. This will be
done iteratively starting from $G$ and following its decomposition in
reverse order.

To capture essential information of the network structure, we 
first define the minimum load $\load$ for which some edge of component 
$C \in \mathcal{C}$ is used in the optimal path $\opt(\load)$.  

\begin{definition} ($\load^{*}_C$).  Given a component $C \in \mathcal{C}$ of a SPG,
we define $\load^{*}_C$ to be the minimum total load such that $C$ is used in the 
optimal path $\opt(\load^{*}_C)$, i.e., if $E_C$ is the set of the edges of $C$, 
$\load^{*}_C = \min\{\load|E_C\cap \opt(\load)\neq \emptyset\}$.
\end{definition}

Next we define the minimum cost required for connecting $\load$ through $C$.

\begin{definition} ($\phi_C$).  We define $\phi_C(\load)$ to be the
  minimum cost for establishing a path that connects a load of $\load$
	from the source, $s_C$, of $C$ to its sink, $t_C$. Note that $\phi_C$ is a strictly concave function as
  the minimum of strictly concave functions\footnote{The cost of each
    path from $s$ to $t$ is a concave function as it is the summation
    of concave functions. Due to Lemma~\ref{lem:optPath}, the optimal
    total cost for each $\load$ equals the minimum cost among those
    paths.}. Trivially, for each edge $e$, $\phi_e=c_e$.
\end{definition}

\begin{remark}
\label{rem:minLoadAndWeightInCs}
If $C$ is constructed by the composition of $C_1$ and
$C_2$, then if this is a
\begin{itemize}
\item parallel composition, $\load^{*}_C \leq \load^{*}_{C_1}$, $\load^{*}_C \leq \load^{*}_{C_2}$, $\phi_C(\load) \leq \phi_{C_1}(\load)$, and $\phi_C(\load) \leq \phi_{C_2}(\load)$ for any $\load$.   
\item series composition, $\load^{*}_C = \load^{*}_{C_1} =\load^{*}_{C_2}$ and $\phi_C(\load)=\phi_{C_1}(\load)+\phi_{C_2}(\load)$ for any $\load$.
\end{itemize}
\end{remark}

We now define $\psi_C$ for each component $C\in\mathcal{C}$ so as to guarantee the following properties: 
\begin{enumerate}
	\item any two components that participate in a parallel composition have the same $\psi$ value. 
	\item the $\psi$ values of any two components that participate in a 
          series composition sum up to the $\psi$ value of the composed
          component. 
	\item $\psi_C > \frac{\phi_C(\load)}{\load}$ for any total load $\load$
          such that $C$ is used in the optimal outcome $\opt(\load)$. 
	\item $\psi_C \leq \phi_C(1)$. 
\end{enumerate}

The first two properties ensure that for each component $C$, we have $\psi_C = \sum_{e\in p}\psi_e$ for any
(acyclic) path $p$ from $s_C$ to $t_C$. The third property guarantees that the highest priority player always 
pays more than the rest of the players. Combined with the fact that all alternative paths have the same $\psi$,
implied by the first two properties, this allows us to show that no two paths can be used in any Nash equilibrium 
(Theorem~\ref{th:concPoA=1}). Finally, the fourth property is crucial in ensuring that the optimal path is a Nash 
equilibrium; if this inequality did not hold, the highest priority player could have an incentive to deviate to the 
path with unit cost $\phi_C(1)$.

\begin{definition}($\psi_C$).
We define the $\psi$ values iteratively, starting from $C_r$ (recall $C_r=G$) and following its decomposition. 
We set $\psi_{C_r}$ to be equal to the minimum total cost when a unit load appears in $G$, i.e. $\psi_{C_r}=\phi_{C_r}(1)$.
Then, after defining a value $\psi_{C}$ for a component $C$, we move on to define the $\psi_{C_1}$ and $\psi_{C_2}$ values 
of the two components $C_1, C_2$ whose composition led to $C$, as follows:
\begin{itemize}
\item If $C$ is constructed by a parallel composition of $C_1$ and $C_2$, $\psi_{C_1} = \psi_{C_2} =\psi_{C}$.
\item If $C$ is constructed by a series composition  of $C_1$ and $C_2$, then we consider two sub-cases:
\begin{itemize}
\item if $\phi_{C_1}(1) < \psi_C\frac{\phi_{C_1}(\load^{*}_C)}{\phi_{C}(\load^{*}_C)}$, then we let $\psi_{C_1}=\phi_{C_1}(1)$ and $\psi_{C_2} = \psi_C - \psi_{C_1}$.
\item if $\phi_{C_2}(1) < \psi_C\frac{\phi_{C_2}(\load^{*}_C)}{\phi_{C}(\load^{*}_C)}$, then we let $\psi_{C_2}=\phi_{C_2}(1)$ and $\psi_{C_1} = \psi_C - \psi_{C_2}$.
\item if $\phi_{C_1}(1) \geq \psi_C\frac{\phi_{C_1}(\load^{*}_C)}{\phi_{C}(\load^{*}_C)}$ and $\phi_{C_2}(1) \geq \psi_C\frac{\phi_{C_2}(\load^{*}_C)}{\phi_{C}(\load^{*}_C)}$, then $\psi_{C_1} = \psi_C\frac{\phi_{C_1}(\load^{*}_C)}{\phi_{C}(\load^{*}_C)}$ and $\psi_{C_2} = \psi_C\frac{\phi_{C_2}(\load^{*}_C)}{\phi_{C}(\load^{*}_C)}$. 
\end{itemize}
Note that the $\psi_{C_1}$ and $\psi_{C_2}$ values as defined above always guarantee that $\psi_{C_1}+\psi_{C_2}=\psi_{C}$. 
To see this note that $\phi_{C_1}(\load)+\phi_{C_2}(\load)=\phi_{C}(\load)$ in a series composition (Remark~\ref{rem:minLoadAndWeightInCs}).
\end{itemize}
\end{definition}

In the next lemma (whose proof can be found in the appendix) we show that $\psi_C \leq \phi_C(1)$ which, as
we mentioned above, is crucial for the optimal path to be a Nash
equilibrium.

\begin{lemma}
\label{lem:psi<=c1}
For any $C\in \mathcal{C}$,  $\psi_C\leq \phi_C(1)$.
\end{lemma}

In the next lemma (whose proof can be found in the appendix) we show that $\psi_C > \frac{\phi_C(\load)}{\load}$ for any $\load$ that $C$ is used in the optimal path.  

\begin{lemma}
\label{lem:psi}
For any $C\in \mathcal{C}$ and any $\load \geq \load^{*}_C$ with $\load >1$, $\psi_C > \frac{\phi_C(\load)}{\load}$; if $\load^{*}_C=1$ then $\psi_C = \phi_C(1)$.
\end{lemma}

In the next lemma (whose proof can be found in the appendix) we formally show that for each edge $e$, the highest
priority player pays at least $\psi_e$ and the rest of the players pay
{\em strictly} less than $\psi_e$. This lemma will be used in a key
argument in our main theorem (Theorem~\ref{th:concPoA=1}) in order to
show that no two paths can be used in any Nash equilibrium.
\begin{lemma}
\label{lem:costShares}
For any strategy profile $\s$ and any edge $e$, $\xi_{ie}(\s) \geq \psi_e$ if $i=h_e(\s)$ and $\xi_{ie}(\s) < \psi_e$ otherwise.
\end{lemma}

In the next two lemmas we show that the protocol belongs in the class of budget-balanced and stable protocols.

\begin{lemma}
	\label{lem:budgetBalance}
	This static-share leader-based cost-sharing protocol is budget-balanced.
\end{lemma}

\begin{proof}
	Consider any strategy profile $\s$ and any edge $e$. If $e\notin \opt(\load_e(\s))$ then the leader,  $h_e(\s)$, covers the whole cost and the rest of the shares are $0$; therefore the protocol is budget-balanced for those edges. If $e\in \opt(\load_e(\s))$ and $\load_e(\s)>1$, by Lemma~\ref{lem:psi<=c1}, the leader covers at most $c_e(1)$ (recall that $\phi_e=c_e$) and the rest of the players cover the rest of the cost by definition. If $e\in \opt(\load_e(\s))$ and $\load_e(\s)=1$, by Lemma~\ref{lem:psi}, the player is charged exactly $c_e(1)$. 
\end{proof}

\begin{lemma}
\label{lem:concStable}
This static-share leader-based cost-sharing protocol is stable.
\end{lemma}

\begin{proof}
  We show that for any load $\load$, the profile where all players
  choose $\opt(\load)$ is a pure Nash equilibrium. If $\load =1$
  then the optimal path has the minimum total cost and therefore that
  player has no incentive to deviate. For the rest of the proof we
  assume that $\load >1$.

Consider any alternative path $p$ and let $C\in \mathcal{C}$ be any component composed by a parallel composition of $C_1$ and $C_2$ such that $\opt(\load)$ uses $C_1$ and $p$ uses $C_2$; let $\opt_C(\load)$ and $p_C$ be the subpaths of $\opt(\load)$ and $p$ inside $C$, respectively\footnote{There should be such a component with source vertex the first vertex that the two paths split and sink vertex the first vertex that they merge again.}. If any player $i$ deviates to $p$, she would be alone at $p_C$ and her cost share would be at least $\phi_C(1)\geq \psi_C$ (by Lemma~\ref{lem:psi<=c1}). However, by Lemma~\ref{lem:costShares},   everybody pays at most $\sum_{e\in \opt_C(\load)} \psi_e = \psi_C$. If we aggregate the cost shares over all such $C$ (those components should be disjoint) we conclude that no player can improve their cost share by choosing $p$ instead of $\opt(\load)$. 
\end{proof}

\begin{theorem}
\label{th:concPoA=1}
This static-share leader-based cost-sharing protocol has PoA $=1$ for series-parallel graphs with concave cost functions.
\end{theorem}
\begin{proof}
We will show that the only possible equilibrium is an optimal path. Consider any strategy profile $\s$ of total load $\load$ that is not an optimum and for the sake of contradiction assume that it is a Nash equilibrium. 

{\bf $\s$ is a single path.} If $\s$ is a single path $p$, there should be some component $C\in \mathcal{C}$ composed by a parallel composition of $C_1$ and $C_2$ such that $\opt(\load)$ uses $C_1$ and $p$ uses $C_2$; let $\opt_C(\load)$ and $p_C$ be the subpaths of $\opt(\load)$ and $p$ inside $C$, respectively. If there are more than one such components, we choose one with $\sum_{e\in p_C}c_e(\load) > \phi_C(\load)$; there should be such a component, otherwise $\s$ is optimum.

Let $h$ be the highest priority player overall. For every $e\in p_C$, it holds that $e\notin \opt(\load)$ and therefore the share of $h$ equals to $\sum_{e\in p_C} c_e(\load) > \phi_C(\load) = \sum_{e\in p^*_C} c_e(\load)\geq \sum_{e\in \opt_C(\load)} c_e(1)$. So, player $h$ has an incentive to deviate from $p_C$ to $\opt_C(\load)$ and hence $\s$ could not be a Nash equilibrium.

{\bf $\s$ is not a single path.} Suppose now that $\s$ is not a single path. There should be some component $C$ composed by a parallel composition of $C_1$ and $C_2$ such that some nonzero load uses both $C_1$ and $C_2$ and it uses only single paths $p_1$ and $p_2$ respectively\footnote{If for any $C_1$ or $C_2$ the load doesn't use a single path we look at the subcomponent where the structure splits until we find the required component.}. Let $\load_1$ and $\load_2$ be the loads using $p_1$ and $p_2$ respectively and let $h_1$ and $h_2$ be the highest priority players in $p_1$ and $p_2$, respectively. W.l.o.g. suppose that $h_1>h_2$. 

By Lemma~\ref{lem:costShares} $h_2$ is charged at least $\psi_e$ for edge $e$ and therefore, he is charged at least $\psi_C$ for the whole path $p_2$. If $h_2$ deviates to $p_1$ he will not be the highest priority player and, by Lemma~\ref{lem:costShares}, for any edge $e\in p_1$ he would be charged strictly less than $\psi_e$ and therefore strictly less than $\psi_C$ for  the whole path $p_1$.Hence, $h_2$ has an incentive to deviate from $p_2$ to $p_1$, meaning that $\s$ could not be a Nash equilibrium.

Overall, we ended up with a contradiction showing that $\s$ cannot be a Nash equilibrium if it is not an optimum. Hence, PoA $=1$.  
\end{proof}

\section{Multicast Games with Concave Cost Functions}\label{sec:multicast}
In this section, we show that our positive results for symmetric games with concave
cost functions cannot be extended to the non-symmetric case of multicast network games, 
i.e., games where each agent $i$ may have a different source $s_i$ that she needs to connect 
to the designated sink $t$. In particular we show that even for constant cost functions, 
which are a special class of concave cost functions, resource-aware protocols cannot achieve 
a constant price of anarchy. Specifically the price of anarchy of any resource-aware protocol
grows with the number of agents $n$ in the instance: we first show a linear lower bound 
for budget-balanced resource-aware protocols, and then we extend it to protocols that allow 
overcharging, obtaining a lower bound of $\sqrt{n}$.

\begin{figure}[h]
	\centering
	\tikzset{middlearrow/.style={
			decoration={markings,
				mark= at position 0.5 with {\arrow{#1}} ,
			},
			postaction={decorate}
		},
	dots/.style={decoration={markings, mark=between positions 0 and 1 step 10pt with { \draw [fill] (0,0) circle [radius=1.0pt];}},
			postaction={decorate}
	}
	}
	\begin{tikzpicture}[xscale=1.7,yscale=1.7]
	
	\node[draw,circle,thick,minimum width =0.8 cm] at (2,2) (t)  {{\large{$\mathbf{t}$}}} ;
	\node[draw,circle,thick,minimum width =0.8 cm] at (0,1) (s1)  {{\large{$\mathbf{s_1}$}}} ;
	\node[draw,circle,thick,minimum width =0.8 cm] at (1,1) (s2)  {{\large{$\mathbf{s_2}$}}} ;
	\node[draw,circle,thick,minimum width =0.8 cm] at (2,1) (s3)  {{\large{$\mathbf{s_3}$}}} ;
	\node[draw,circle,thick,minimum width =0.8 cm] at (4,1) (sn)  {{\large{$\mathbf{s_{n}}$}}} ;
	\node[draw,circle,thick,minimum width =0.8 cm] at (2,0) (v)  {{\large{$\mathbf{v}$}}} ;

	\path[dots] (2.5,1) to (3.5,1);
	
\draw[thick, middlearrow={stealth}] (s1) to node [above] {{\large{$\mathbf{1}$}}} (t);
\draw[thick, middlearrow={stealth}] (s1) to node [below] {{\large{$\mathbf{0}$}}} (v);
\draw[thick, middlearrow={stealth}] (s2) to node [above] {{\large{$\mathbf{1}$}}} (t);
\draw[thick, middlearrow={stealth}] (s2) to node [below] {{\large{$\mathbf{0}$}}} (v);
\draw[very thick, middlearrow={stealth}] (s3) to node [left] {{\large{$\mathbf{1}$}}} (t);
\draw[very thick, middlearrow={stealth}] (s3) to node [left] {{\large{$\mathbf{0}$}}} (v);
\draw[thick, middlearrow={stealth}] (sn) to node [above] {{\large{$\mathbf{1}$}}} (t);
\draw[thick, middlearrow={stealth}] (sn) to node [below] {{\large{$\mathbf{0}$}}} (v);
\draw[thick, middlearrow={stealth}] (v.east) to [out=0,in=-90]  ([xshift={0.5cm}]sn.east) node[right]{\large{$\mathbf{c}$}} to [out=90, in=0] (t.east) ;
	\end{tikzpicture}
	\caption{A simple acyclic graph with a sink $t$ and sources $s_1, s_2, \dots, s_n$ for each agent. All of the cost functions are constant, i.e., the cost of the edge is zero if it is not used and some constant if it is used. The direct edges from the sources to the sink have a constant cost of 1 and those from the sources to $v$ have a constant cost of 0. The cost of the edge connecting $v$ to the sink is some constant $c$, which we define appropriately for the proofs of Theorems \ref{thm:multicastLB1} and \ref{thm:multicastLB2}.}
	\label{fig:multiBB}
\end{figure}
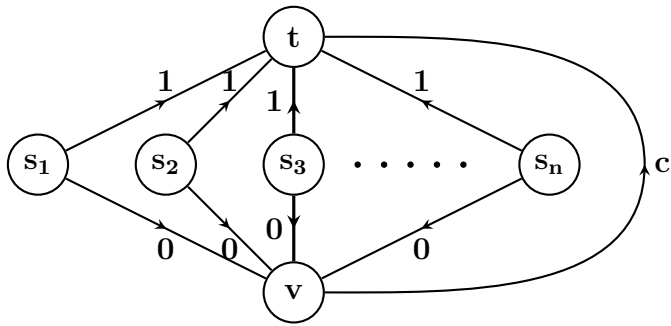

\begin{theorem}\label{thm:multicastLB1}
  There is no resource-aware budget-balanced protocol that can achieve
  a PoA better than ${n}$ for the case of multicast networks with constant cost functions.
\end{theorem}

The proof follows directly from the instance of Proposition 4.12 of
  \cite{CRV10} and the instance is presented in Figure~\ref{fig:multiBB} where $c$ should be $1$. We give the complete proof in the appendix.

\begin{theorem}\label{thm:multicastLB2}
  There is no resource-aware protocol even with overcharging that can
  achieve a PoA better than $\sqrt{n}$ for the case of multicast networks
	with constant cost functions.
\end{theorem}
\begin{proof}
We again consider the graph of $n+2$ vertices from Figure~\ref{fig:multiBB}, as in the previous proof, 
with the only difference that the constant cost $c$ of the $(v, t)$ edge is set to be 
equal to $\sqrt{n}$. Note that this cost function is constant
and does not depend on the number of agents using it (its cost is zero if
no agents use it and equal to $\sqrt{n}$ if at least one agent uses i).
If the resource-aware protocol uses overcharging but the new costs of the
the $(s_i, t)$ edges after the overcharging are no more than $\sqrt{n}$ for all $i$, 
then we consider the instance with $n$ agents with sources $s_1, s_2, \dots, s_n$ and
claim that the strategy profile where every one of these agent uses their $(s_i, t)$ 
edge is an equilibrium: a unilateral deviation through edge $(v, t)$ would cost at least 
$\sqrt{n}$ due to the cost of $(v, t)$, while every agent's cost including the overcharging
is at most $\sqrt{n}$.
On the other hand, if the resource-aware protocol increases the cost of an $(s_i, t)$ 
edge above $\sqrt{n}$ for some $i$, then consider the alternative problem instance 
with the same graph, but with only one player, player $i$, in the system. Since the 
resource-aware
protocol that decides the cost of edge $(s_i, t)$ cannot tell the difference between 
these two instances, it would still increase the cost of that edge to above $\sqrt{n}$
for this instance and force the agent to pay a cost of $\sqrt{n}$. On the other hand, 
in the optimal solution for that instance agent $i$ should just use the $(s_i, t)$ edge and pay a cost of 1. \end{proof}

\section{Convex Cost Functions}\label{sec:convex}
For the rest of the paper, we now consider networks convex cost functions
rather than concave ones. We first prove that for convex cost functions, 
just like in the case of concave cost functions, there exists a budget-balanced
protocol with a PoA of 1 for series-parallel graphs. This protocol is the {\em
  incremental cost-sharing protocol} introduced by Moulin~\cite{Mou99}
and it is the same protocol used in \cite{CGS17} for parallel
links. We note that this is an oblivious protocol, meaning that it
requires no knowledge of the instance other the number of agents using
the edge at hand.

We complement this positive result by showing that in symmetric games
on directed acyclic graphs any resource-aware budget-balanced
protocol has PoA $ =\Omega(n)$. For protocols that use overcharging, 
we show that optimality cannot be achieved, by showing a lower bound of 
$1.18$, and we leave as an open question the existence (or not) of a 
protocol that achieves a constant PoA with the use of overcharging. 
However, if we consider games beyond ones that are symmetric, such as 
multicast network games, we provide a lower bound of $\sqrt{n}$ for 
all resource-aware protocols even with use of overcharging.

\subsection{Series-Parallel Graphs}
In this section we show that in series-parallel graphs (SPGs) with convex cost functions the {\em incremental cost-sharing protocol} proposed by Moulin~\cite{Mou99} has PoA $=1$. The incremental cost-sharing protocol considers a global order $\pi$ of the players and defines the cost-share of each player $i$ for using edge $e$ to be its marginal contribution if only players preceding him in $\pi$ were using $e$. For simplicity, for the rest of the section we name the players based on the order $\pi$. 
\begin{definition}(Prior Load $\load_e^{< i}(\s)$).
	Given a strategy profile $\s$ and an edge $e$, we define the {\em prior load}, $\load_e^{< i}(\s)$, 
	for a player $i$ using $e$ to be the load on $e$ due to players preceding $i$ according to $\pi$, i.e.,
	\[\load_e^{< i}(\s) = |\{k<i: e\in \str_k\}|\,. \]
\end{definition}

Based on the definition of the prior load we can formally define the incremental cost-sharing protocol as follows.
\begin{definition}(Incremental Cost-Sharing Protocol).
 Given a strategy profile $\s$, the cost share of player $i$ for using edge $e$ is 
\[\xi_{ie}(\s) = c_e(\load_e^{< i}(\s)+1) - c_e(\load_e^{< i}(\s))\,.\]
\end{definition}

The following lemma is a key lemma in order to show that the incremental cost-sharing protocol has PoA $=1$. Due to space limitations, the proof of this lemma can be found in the appendix.

\begin{lemma}
	\label{lem:NoPlayersInOpt}
Given any Nash equilibrium $\s$, there exists an optimal assignment $\s^*$ such that for any player $i$ and any edge $e\in \str^*_i$, it holds that $\load_e^{< i}(\s) \leq \load_e^{< i}(\s^*)$.
\end{lemma}

Next we give a technical lemma to be used in our main theorem. 

\begin{lemma}
	\label{lem:techConcave}
	Let $\s=(\str_1,\ldots ,\str_n)$ be any Nash equilibrium under the incremental cost-sharing protocol and $\s^*=(\str^*_1,\ldots , \str^*_n)$ be some optimal assignment satisfying Lemma~\ref{lem:NoPlayersInOpt}. Then for any player $i$ and any $e\in \str^*_i$ it holds that 
	$$\xi_{ie}(\str^*_i,\s_{-i})
	\leq  \xi_{ie}(\s^*).$$
\end{lemma}

\begin{proof}
By the definition of the incremental cost-sharing protocol	
	\begin{eqnarray*}
	\xi_{ie}(\str^*_i,\s_{-i}) &=& c_e(\load_e^{< i}(\str^*_i,\s_{-i})+1) - c_e(\load_e^{< i}(\str^*_i,\s_{-i}))\\
	&=&c_e(\load_e^{< i}(\s)+1) - c_e(\load_e^{< i}(\s)) \qquad\qquad\;\;\; \mbox{(by the definition of $\load_e^{< i}(\s)$)}\\
	&\leq& c_e(\load_e^{< i}(\s^*)+1) - c_e(\load_e^{< i}(\s^*)) \qquad\qquad \mbox{(due to Lemma \ref{lem:NoPlayersInOpt} and convexity)}\\
	&=& \xi_{ie}(\s^*).
	\end{eqnarray*}
\end{proof}

\begin{theorem}
The incremental cost-sharing protocol has PoA $ = 1$ in series-parallel graphs with convex cost function.
\end{theorem}

\begin{proof}

Let $\s=(\str_1,\ldots ,\str_n)$ be any Nash equilibrium under this protocol and let $\s^*=(\str^*_1,\ldots , \str^*_n)$ be some optimal assignment satisfying Lemma~\ref{lem:NoPlayersInOpt}. Then,

\begin{eqnarray*}
C(\s) &=& \sum_i \sum_{e\in \str_i} \xi_{ie}(\s)
\leq \sum_i \sum_{e\in \str^*_i} \xi_{ie}(\str^*_i,\s_{-i})
\leq \sum_i \sum_{e\in \str^*_i} \xi_{ie}(\s^*)
= C(\s^*),
\end{eqnarray*}
where the first inequality is due to the fact that $\s$ is a Nash equilibrium and the second inequality comes from Lemma~\ref{lem:techConcave}.
\end{proof}

\subsection{Directed Acyclic Graphs} 
In this section we consider symmetric games on directed acyclic graphs
and prove a lower-bound of $\Omega(n)$ on the PoA of all budget
balanced resource-aware protocols. Then, we obtain a lower-bound of
$1.18$ for protocols that allow overcharging.

\begin{theorem}
	\label{poadagconvex}
	Any stable budget-balanced resource-aware cost-sharing protocol has PoA $=\Omega(n)$ for directed acyclic graphs with convex cost functions.
\end{theorem}

\begin{proof}
		We use the DAG shown in Figure \ref{fig:poaexample} to prove the statement. In this graph all the players want to go from $s$ to $t$ and the cost of each edge is either $0$ or $1$ if a single player uses the edge and infinity if more than one players use it.
		
		\begin{figure}[h]
			\centering
			\tikzset{middlearrow/.style={
					decoration={markings,
						mark= at position 0.5 with {\arrow{#1}} ,
					},
					postaction={decorate}
				},
				dots/.style={decoration={markings, mark=between positions 0 and 1 step 10pt with { \draw [fill] (0,0) circle [radius=1.0pt];}},
					postaction={decorate}
				}
			}
			\begin{tikzpicture}[xscale=0.8,yscale=0.8]
			\node[draw,circle,thick,minimum width =1 cm] at (0,8) (s) {$\mathbf{s}$} ;
			\node[draw,circle,thick,minimum width =1 cm] at (-5.5,5) (v1) {$\mathbf{u_n}$} ;
			\node[draw,circle,thick,minimum width =1 cm] at (-3.5,5) (v2) {$\mathbf{v_n}$} ;
			\node[draw,circle,thick,minimum width =1 cm] at (-1.5,5) (v3) {$\mathbf{u_{n-1}}$} ;
			\node[draw,circle,thick,minimum width =1 cm] at (1.5,5) (v4) {$\mathbf{v_2}$} ;
			\node[draw,circle,thick,minimum width =1 cm] at (3.5,5) (v5) {$\mathbf{u_1}$} ;
			\node[draw,circle,thick,minimum width =1 cm] at (5.5, 5) (v6) {$\mathbf{v_1}$} ;
			\node[draw,circle,thick,minimum width =1 cm] at (0, 2) (t) {$\mathbf{t}$} ;
			
			\path[dots] (0.45,5) to (-0.55,5);
			
			\draw[thick, middlearrow={stealth}] (s) to node [above, sloped] {{$0, \infty$}} (v2);
			\draw[thick, middlearrow={stealth}] (v1) to node [above, sloped] {{$0, \infty$}} (t);
			\draw[thick, middlearrow={stealth}] (s) to node [above, sloped] {{$0, \infty$}} (v4);
			\draw[thick, middlearrow={stealth}] (v3) to node [above, sloped] {{$0, \infty$}} (t);
			\draw[thick, middlearrow={stealth}] (s) to node [above, sloped] {{$0, \infty$}} (v6);
			\draw[thick, middlearrow={stealth}] (v5) to node [below, sloped] {{$0, \infty$}} (t);
			
			\draw[thick, middlearrow={stealth}] (v6) to node [below, sloped] {{$1, \infty$}} (v5);
			\draw[thick, middlearrow={stealth}] (v4) to node [above, sloped] {{$0, \infty$}} (v5);
			\draw[thick, middlearrow={stealth}] (v2) to node [above, sloped] {{$1, \infty$}} (v1);	
			\draw[thick, middlearrow={stealth}] (v2) to node [below, sloped] {{$0, \infty$}} (v3);	
			
			\draw[thick, middlearrow={stealth}] (s.east) to [out=0,in=90]  ([xshift={1cm}]v6.east) node[right]{{$1, \infty$}} to [out=-90, in=0] (t.east);

			\end{tikzpicture}
			\caption{An acyclic graph with a single source $s$ and sink $t$. All of the cost functions are infinity for load $\load\geq 2$. This graph is used to prove Theorem~\ref{poadagconvex}.}
			\label{fig:poaexample}
		\end{figure}
		
		Considering $n+1$ players, in the optimum, one player uses the $(s,t)$ edge, and the others use  $\mathbf{sv_{i}u_{i}t}$ paths. The total cost of the optimum is $n+1$. If any other strategy profile was an equilibrium it would result in an unbounded total cost and give an unbounded PoA. Therefore, there should exist an optimum strategy profile $\s^*$ that is a Nash equilibrium. Assume w.l.o.g. that in $\s^*$, player $n+1$ uses the $(s,t)$ edge, and for each $i \leq n$, player $i$ uses the $\mathbf{sv_{i}u_{i}t}$ path. 
		
		If we now consider the instance in which only the first $n$ players use the network, the strategy profile in which player $i$ uses the $\mathbf{sv_{i}u_{i}t}$ path should also be an equilibrium. The reason is that since $\s^*$ is a Nash equilibrium, no player has an incentive to deviate to any other path but the $(s,t)$ edge, and obviously, no player would prefer the $(s,t)$ edge because she pays the same cost. The cost of this strategy profile is $n$. However, the optimum with $n$ players is for one player to use the $(s,t)$ edge, and the others to use the $\mathbf{sv_{i}u_{i-1}t}$, for $i\geq 2$ with total cost $1$. This results in PoA $=\Omega(n)$. \end{proof}

In the following theorem we give a lower-bound of $1.18$ for protocols that allow overcharging. The proof can be found in the appendix.

\begin{theorem}
	 \label{poadagconvexovercharging}
	There is no stable resource-aware cost-sharing mechanism with $\text{PoA} < \frac{\sqrt{33}-1}{4} \simeq 1.18$, for directed acyclic graphs with convex cost functions even with overcharging.
\end{theorem}

\subsection{Multicast Games}
Here we extend our lower-bounds to multicast games. We prove a lower-bound of $\sqrt{n}$ even  for protocols that allow overcharging.

\begin{theorem}
	\label{multidagconvex}
	There is no resource-aware protocol even with overcharging that can
	achieve a PoA better than $\sqrt{n}$ for the case of multicast networks
	with convex cost functions.
\end{theorem}

This lower bound is inspired by the lower bound we gave in Theorem~\ref{poadagconvex}. The difference is that the cost of the $(s,t)$ edge is $\sqrt{n}$ for a single player and we use two sources $s_1$ and $s_2$, where from $s_2$ there are only the alternative paths, namely $(s_2,t)$ and  $\mathbf{s_2v_1u_1t}$ and all other paths start from $s_1$. Then the costly equilibria appeared in Theorem~\ref{poadagconvex} can only be avoided by charging the player using path $\mathbf{s_2v_1u_1t}$ at least $\sqrt{n}$. Then if only that player appears in the system, he cannot avoid the $\sqrt{n}$ charge whereas the original cost of the optimum was $1$. We give the complete proof in the appendix.

	\section*{Acknowledgments}
	The work of the second author was supported by NSF grant CCF-1755955. Part of this work was done when the third and fourth authors were employed by the Max-Planck-Institute for Informatics, Saarland University Campus, Saarbr\"ucken, Germany. The fourth author was supported by the Lise Meitner Award Fellowship.

	\bibliographystyle{plainnat}
	\bibliography{cost-sharing}
	
	\newpage
	\appendix
	\section{Proofs Deferred from the Main Body of the Paper}

\subsection{Proof of Lemma~\ref{lem:optPath}}
\begin{proof}
	Consider any optimal allocation with at least two paths $p_1, p_2$
	from $s$ to $t$; if there is no such allocation then all optimal
	allocations use single paths and the lemma follows. We next show
	that shifting one player either from $p_1$ to $p_2$ or the other way
	around does not increase the total cost.
	
	Let $\load_e$ be the load at each edge $e$ under the assumed optimal allocation. Suppose w.l.o.g. that 
	\begin{equation}
	\sum_{e\in p_1}(c_e(\load_e)-c_e(\load_e-1))\geq \sum_{e\in p_2}(c_e(\load_e)-c_e(\load_e-1))\,. \label{ComparingMarginals}
	\end{equation}
	By cancelling out common edges and using the concavity of the
	cost functions we get,
	\begin{equation}
	\sum_{e\in p_1\smallsetminus p_2}(c_e(\load_e)-c_e(\load_e-1))\geq \sum_{e\in p_2 \smallsetminus p_1}(c_e(\load_e+1)-c_e(\load_e))\,. \label{optProp}
	\end{equation}
	
	Next we show that shifting one player from $p_1$ to $p_2$
	will not increase the total cost of the edges belonging to
	$p_1\cup p_2$, hence the overall cost will not increase since
	the load on the rest of the edges remains intact. The total
	cost of the edges in $p_1\cup p_2$ after a player shifts from
	$p_1$ to $p_2$ is equal to
	\begin{eqnarray*}
		&&\sum_{e\in p_1\smallsetminus p_2}c_e(\load_e - 1) +
		\sum_{e\in p_1\cap p_2}c_e(\load_e) +
		\sum_{e\in p_2\smallsetminus p_1}c_e(\load_e + 1)\\
		&\leq& 	\sum_{e\in p_1\smallsetminus p_2}c_e(\load_e) +
		\sum_{e\in p_2 \smallsetminus p_1}c_e(\load_e) +
		\sum_{e\in p_1\cap p_2}c_e(\load_e) = \sum_{e\in p_1\cup p_2}c_e(\load_e)\,, \qquad \mbox{(by using \eqref{optProp})}
	\end{eqnarray*} 
	where $\sum_{e\in p_1\cup p_2}c_e(\load_e)$ is the current
	total cost of the edges in $p_1\cup p_2$. Next we show
	that \eqref{ComparingMarginals} still holds for the new
	allocation. This implies that if all the players are shifted
	from $p_1$ to $p_2$ the cost will not increase and therefore
	we can construct an optimal allocation with one less path. 
	Working that way, we will end up with a single path.
	
	Let $\load_e'$ be the new load on edge $e$ after shifting a
	player from $p_1$ to $p_2$. Obviously, $\load_e'\leq \load_e$
	for any $e \in p_1$ and $\load_e'\geq \load_e$ for any $e \in
	p_2$. Then, due to concavity, $$\sum_{e\in
		p_1}(c_e(\load_e')-c_e(\load_e'-1)) \geq \sum_{e\in
		p_1}(c_e(\load_e)-c_e(\load_e-1)) \geq \sum_{e\in
		p_2}(c_e(\load_e)-c_e(\load_e-1)) \geq \sum_{e\in
		p_2}(c_e(\load_e')-c_e(\load_e'-1))\,.$$ Therefore,
	\eqref{ComparingMarginals} holds for the new allocation and
	the lemma follows.
\end{proof} 

\subsection{Proof of Lemma~\ref{lem:BreakTies}}
\begin{proof}
	Let $e_1,\ldots e_k$ be all the edges and $r$ be an arbitrarily big constant. At first, we round up all the costs to $r$ decimal places, i.e., for each edge $e$ and load $\load$, 
	we round up $c_e(\load)$ to the closest multiple of $10^{-r}$. Then we increase $c_{e_i}(\load)$ by $\plu! 10^{-(r+\plu! i)}$, where $\plu$ is the maximum number of players ever appears. After these increments, for any $e_i$, any cost share $\zeta_{e}(\load)$ has some non-zero values at decimal places from $r+\plu! (i-1)+1$ to $r+\plu! i$ and the rest of the decimal places greater than $r$ have zero values. 
	Since each edge appears at most once in a path, the cost shares of two different paths are equal if and only if they have the same set of edges.
	
	Note that the increment at each edge is upper bounded by $10^r$ and therefore can be arbitrarily small as it only depends on the choice of $r$. 
\end{proof}

\subsection{Proof of Theorem~\ref{thm:LB_SSLBprotocols}}
\begin{proof}
	In order to show this lower bound we construct a graph with appropriate edge cost functions. Consider the graph of Figure~\ref{fig:BB_LB} which is basically a Braess's network with edge $(s,u)$ replaced by multiple edges $e_1,\ldots ,e_r$. The edges $(v,u)$ and $(u,t)$ have constant cost of $0$ and $2k$, respectively, for some integer $k\geq 6$. The cost of the edge $(v,t)$ equals the load on that edge. For the rest of the edges we define the cost functions as follows:
	
	\[
	c_0(\load) =
	\begin{cases}
	\load & \text{if } \load \leq k\\
	k & \text{otherwise}
	\end{cases}\,,
	\qquad\qquad c_j(\load) = \frac{\load}{jk^2}+\frac{H_{j-1}}k+\varepsilon_j\,,  \quad \text{for }1\leq j\leq r   \,,
	\]
	where $H_j$ is the harmonic number, i.e. $H_j=\sum_{i=1}^j1/i$, for $j\geq 1$ and $H_0=0$, and $0< \varepsilon_1 < \varepsilon_2 < \ldots < \varepsilon_r <\frac 1{r^2k^2}$ are used to guarantee a single optimum for each $\load$. $r$ should be some integer such that $c_r(1)<1$; we set $r=2^{k}$ that satisfies this inequality if someone notice that $H_{2^k-1}\leq k -1$ for $k\geq 6$.\footnote{We remark that our result holds for strictly concave functions as well. To see this, one can add $\epsilon - \epsilon/\load$ to all cost functions and verify that all arguments still hold for sufficiently small $\epsilon>0$.}

	\begin{figure}[h]
		\centering
		\tikzset{middlearrow/.style={
				decoration={markings,
					mark= at position 0.5 with {\arrow{#1}} ,
				},
				postaction={decorate}
			},
			dots/.style={decoration={markings, mark=between positions 0 and 1 step 10pt with { \draw [fill] (0,0) circle [radius=1.0pt];}},
				postaction={decorate}
			},
			el/.style = {inner sep=2pt, align=left, sloped}
		}
		\begin{tikzpicture}[xscale=1,yscale=1]

		\node[draw,circle,thick,fill=black,inner sep=1.4pt,radius=0.25pt,label={[label distance=-1pt]below:{$s$}}] at (0,0) (s) {} ;
		\node[draw,circle,thick,fill=black,inner sep=1.4pt,radius=0.25pt,label={[label distance=1pt]above:{$u$}}] at (5,2) (u) {} ;
		\node[draw,circle,thick,fill=black,inner sep=1.4pt,radius=0.25pt,label={[label distance=-1pt]below:{$v$}}] at (5,-2) (v) {} ;
		\node[draw,circle,thick,fill=black,inner sep=1.4pt,radius=0.25pt,label={[label distance=-1pt]below:{$t$}}] at (10,0) (t) {} ;
		
		\draw[very thick, middlearrow={stealth}] (s) to node [el,below] {{\large{$\mathbf{c_0(\load)}$}}}  (v);
		\draw[very thick, middlearrow={stealth}] (s) to node [el,below] {{\large{$\mathbf{c_2(\load)}$}}} node[el,below,pos=0.8] {{\large{$\mathbf{e_2}$}}} (u);
		\draw[very thick, middlearrow={stealth}] (v) to node [right] {{\large{$\mathbf{0}$}}} (u);
		\draw[very thick, middlearrow={stealth}] (v) to node [el,below] {{\large{$\mathbf{\load}$}}}  (t);
		\draw[very thick, middlearrow={stealth}] (u) to node [el,above] {{\large{$\mathbf{2k}$}}} (t);
		
		\draw[thick, middlearrow={stealth}] (s) to [bend left=20] node [el,above] {{$\mathbf{c_3(\load)}$}} node[el,above,pos=0.8] {{\large{$\mathbf{e_3}$}}}(u);
		\draw[thick, middlearrow={stealth}] (s) to [bend left=-30] node [el,below] {{$\mathbf{c_1(\load)}$}} node[el,below,pos=0.8] {{\large{$\mathbf{e_1}$}}} (u);
		\draw[thick, middlearrow={stealth}] (s) to  [bend left=90] node [el,above] {{$\mathbf{c_r(\load)}$}} node[el,above,pos=0.8] {{\large{$\mathbf{e_r}$}}} (u);
		
		\path[dots] (1.8,1.5) to (1.3,2.3);
		
		\end{tikzpicture}
		\caption{Graph where PoA $=\Omega\left(\frac{\sqrt{n}}{\log^2n}\right)$ for all static-share leader-based protocol.}
		\label{fig:BB_LB}
	\end{figure}
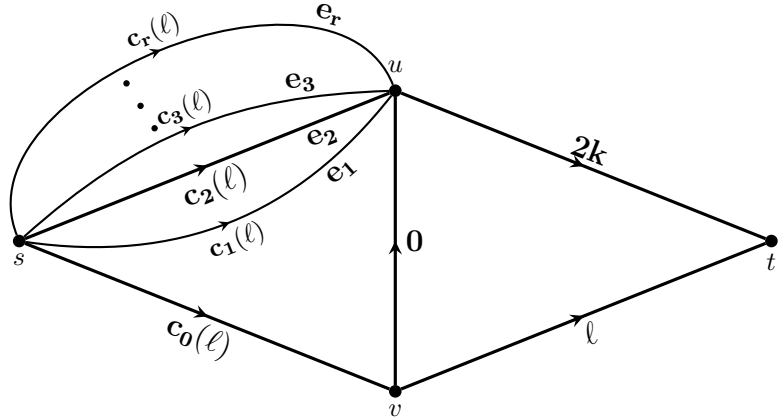
	
	For simplicity we denote by $sv$, $vu$, $ut$ and $vt$, the edges $(s,v)$, $(v,u)$, $(u,t)$ and $(v,t)$, respectively. 
	Moreover, by $svt$ we denote the path using edges $sv$ and $vt$ and similarly for the rest of the paths. 
	
	First note that, for $\load \leq k$, $c_0(\load)+\load \leq 2k$ meaning that the only optimal path for that load is the $svt$. It is not hard to verify that 
	\begin{eqnarray}
	c_0(\load)+\load > c_1(\load) +2k\,, &&\text{ for } \load \geq k+1\,,\label{eq;optNotsvt}\\
	c_j(\load)< c_{j+1}(\load)\,, &&\text{ for } \load\leq (j+1)k\,,\label{eq:optOfC1}\\ 
	c_{j-1}(\load)> c_{j}(\load)\,, &&\text{ for } \load\geq jk+1\,, \label{eq:optOfCr}
	\end{eqnarray}
	for all $j\geq 1$. Inequality \eqref{eq;optNotsvt} indicates that for $\load\geq k+1$ the path $svt$ cannot be the optimum. Inequalities \eqref{eq:optOfC1} and \eqref{eq:optOfCr} indicate that among all edges connecting $s$ to $u$, $e_j$ has the minimum cost for $jk+1 \leq \load \leq (j+1)k$. Let $\load^*$ be the minimum load such that $c_r(\load^*) > k$ meaning that for $\load \geq \load^*$ the only optimal path is the $svut$; note that $\load^*>rk+1$. In Table~\ref{tbl:optPaths} we summarize the optimal path for each load. 
	
	\begin{table}[h]
		\begin{center}
			\begin{tabular}{|c||c|}
				\hline 
				$\load$ & $\opt(\load)$\\ \hline \hline
				$[1,k]$ & $svt$  \\ \hline
				$[jk+1,(j+1)k], 1\leq j < r$ & $sut$ through $e_j$  \\ \hline
				$[rk+1,\load^*-1]$ & $sut$ through $e_r$  \\ \hline
				$[\load^*,\infty)$ & $svut$ \\ \hline
				
			\end{tabular}
		\end{center}
		\caption{Optimal paths for the network of Figure \ref{fig:BB_LB}.}
		\label{tbl:optPaths}
	\end{table}
	
	Consider now $n=r^2k^2$ players and any Nash equilibrium $\s$ of any static-share leader-based protocol. We show that PoA $=\Omega\left(\frac{\sqrt{n}}{\log^2n}\right)$.
	
	First note that if more than $n/2$ players use some of the edges connecting $s$ to $u$, due to concavity the minimum cost occurred if all those players use the same edge (Lemma~\ref{lem:optPath}) and due to inequality \eqref{eq:optOfCr} this edge is $e_r$. Therefore, the total cost is at least $c_r(n/2)>r/2$. It is easy to check that $\load^*\leq r^2k^2=n$ (recall that $\load^*$ is the minimum load such that $c_r(\load^*) > k$ and $r=2^k$) and hence the optimum is for all players to use the path $svut$ with total cost $3k$. So, PoA $>\frac r{6k} = \Omega\left(\frac{\sqrt{n}}{\log^2n}\right)$.\footnote{This is true because $\frac{\sqrt{n}}{\log^2 n}=\frac{rk}{(\log 2^{2k}k^2)^2}<\frac{rk}{(\log 2^{k})^2}=\frac{r}{k}.$}
	
	Now for the rest of the proof suppose that at least $n/2$ players use $sv$ under $\s$. We distinguish between three cases based on which path the highest priority player overall, $h$, uses in $\s$. 
	
	\paragraph{{\bf $\mathbf{h}$ uses $\mathbf{svt}$}}: If any edge connecting $s$ to $u$ is used under $\s$, definitely one of its users is charged by at least $1/rk^2$. If that player deviates from that edge to $svu$, he would pay at most $2/r^2k$ as he wouldn't be the highest priority player in $sv$; $h$ is. But, $2/r^2k < 1/rk^2$, for $k\geq 6$, meaning that $\s$ couldn't be a Nash equilibrium. 
	
	Therefore, only the paths $svut$ and $svt$ can be used in $\s$. If $\load_{ut}(\s) \leq k$, then $\load_{vt}(\s) \geq n- k$ resulting in PoA $>n/3k=\Omega(n/\log n)$. 
	So we consider at last the case that $\load_{ut}(\s) > k$ meaning that for any $\load\geq \load_{ut}(\s)$, $ut\in \opt(\load)$. First note that any player but $h$ using $vt$ pays either $1$ or $0$. The reason is that since $vt\in \opt(1)$, the $\psi_{vt}$ defined by the protocol should be $1$ so that it is budget-balanced. Similarly, $\psi_{sv} = 1$. This further means that $h$ is charged at least $2$ for using $svt$. We take two cases on the value of $\psi_{ut}$.
	
	\begin{itemize}
		\item $\psi_{ut}>1$: player $h_{ut}$ pays more than $1$ for the path $vut$ and if he deviates from $vut$ to $vt$ he would pay at most $1$, as we mentioned above.
		\item $\psi_{ut}\leq 1$: player $h$ currently pays at least $2$ and if he deviates from $svt$ to $sut$ he would pay strictly less than $2$, since the unit cost of any edge connecting $s$ to $u$ is less than $1$ and the charge for using $ut$ would be $\psi_{ut}$ since $ut\in \opt(\load_{e_{ut}}(\s)+1)$.
	\end{itemize}
	
	Hence, $\s$ cannot be a Nash equilibrium in this case. Overall, if $h$ uses $svt$, PoA $=\Omega(n/\log n)$.

	\paragraph{{\bf $\mathbf{h}$ uses $\mathbf{sut}$}}: Similarly as above, if $\load_{ut}(\s) \leq 2k+1$, PoA $=\Omega(n/\log n)$, so we assume that $\load_{ut}(\s) \geq 2k+2$. If $vt$ is used under $\s$, then $h_{vt}$ pays at least $1$ for $vt$ whereas if he deviates from $vt$ to $vut$ he would pay less than $1$ since he is not the highest priority player, $h$ is. Therefore, $vt$ cannot be used in $\s$, meaning that $h_{sv}$ follows the path $svut$. It is easy to argue now that every edge connecting $s$ to $u$ is used under $\s$, otherwise $h_{sv}$ would have an incentive to deviate to the empty edge. This is because $h_{sv}$ pays at least $1$ for $sv$ and the unit cost of any $su$ edge is less than $1$. By Claim~\ref{clm:e_jused} PoA $=\Omega\left(\frac{\sqrt{n}}{\log^2n}\right)$.
	
	\paragraph{{\bf $\mathbf{h}$ uses $\mathbf{svut}$}}: In this case we can argue right away that every $su$ edge is used under $\s$, otherwise $h$ would have an incentive to deviate to the empty edge. Then again by Claim~\ref{clm:e_jused} PoA $=\Omega\left(\frac{\sqrt{n}}{\log^2n}\right)$.
	
	\begin{claim}
		\label{clm:e_jused}
		If $n=r^2k^2$ players appear and every $su$ edge is used in some Nash equilibrium $\s$, then PoA $=\Omega\left(\frac{\sqrt{n}}{\log^2n}\right)$.
	\end{claim}
	
	\begin{proof}	
		The total cost of $\s$ is at least 
	$$\sum_{j=1}^{r-1}c_j(1)\geq \frac 1k \sum_{j=0}^{r-2}H_j = \frac 1{\log r} ((r-1)H_{r-2}-(r-2)=\Theta(r)\,.\footnote{\text{We use the property } $\sum_{j=1}^nH_j=(n+1)H_n-n$.}$$

		The optimal total cost is $3k$, which gives PoA $=\Omega\left(\frac{\sqrt{n}}{\log^2n}\right)$.
	\end{proof}
\end{proof}

\subsection{Proof of Lemma~\ref{lem:psi<=c1}}
\begin{proof}
	We are going to prove this lemma by induction starting from $G$. For
	$G$ the statement holds by definition.
	
	Suppose that $\psi_{C}\leq \phi_{C}(1)$ for some $C \in
	\mathcal{C}$. We will show that the statement also holds for its
	components $C_1, C_2$ as well.
	
	\begin{itemize}
		\item If $C$ is constructed by a parallel composition, for both $i
		\in\{1,2\}$, it holds by Remark~\ref{rem:minLoadAndWeightInCs}
		that $$\psi_{C_i}=\psi_{C}\leq \phi_C(1)\leq \phi_{C_i}(1).$$
		\item If $C$ is constructed by a series composition, if the last case holds in the definition of $\psi_{C_i}$, then the statement trivially holds. 
		If w.l.o.g. the first case holds then 
		$$\psi_{C_2} = \psi_{C} - \psi_{C_1} \leq  \phi_{C}(1) - \phi_{C_1}(1) = \phi_{C_{2}}(1),$$
		where the inequality holds by assumption and the last equality is due to Remark~\ref{rem:minLoadAndWeightInCs}.
	\end{itemize}
\end{proof}

\subsection{Proof of Lemma~\ref{lem:psi}}
\begin{proof}
	We also prove this lemma by induction starting from $G$. For $G$ by definition $\psi_{C_r} = \phi_{C_r}(1)$ and since $\phi_{C_r}$ is a strictly concave function $\psi_{C_r} > \frac{\phi_{C_r}(\load)}{\load}$ for any $\load>1$.
	
	Suppose that the statement holds for some $C\in \mathcal{C}$. We show that the statement also holds for its components $C_1, C_2$ as well.
	
	\begin{itemize}
		\item If $C$ is constructed by a parallel composition, for any $i\in\{1,2\}$, we consider two cases:
		\begin{itemize}
			\item $\load^{*}_{C_i}=1$: by Remark~\ref{rem:minLoadAndWeightInCs}, $\load^{*}_{C}=1$ and further $\phi_C(1)=\phi_{C_i}(1)$. By the definition of $\psi_{C_i}$ and the assumption that $C$ satisfies the lemma's statement, $$\psi_{C_i}=\psi_{C}=\phi_C(1)=\phi_{C_i}(1)\,.$$ 
			\item $\load^{*}_{C_i}>1$: for any $\load\geq \load^{*}_{C_i}$, 
			$$\psi_{C_i}=\psi_{C} > \frac{\phi_C(\load^{*}_{C_i})}{\load^{*}_{C_i}}=\frac{\phi_{C_i}(\load^{*}_{C_i})}{\load^{*}_{C_i}}\geq \frac{\phi_{C_i}(\load)}{\load}\,,$$
			where the first inequality is because $C$ satisfies the lemma's statement for $\load^{*}_{C_i} \geq \load^{*}_C$ and $\load^{*}_{C_i}>1$ ($\load^{*}_{C_i} \geq \load^{*}_C$ is true due to Remark~\ref{rem:minLoadAndWeightInCs}). For the second equality note that if $\phi_C(\load^{*}_{C_i}) < \phi_{C_i}(\load^{*}_{C_i})$ the optimal path wouldn't go through $C_i$, which contradicts the definition of $\load^{*}_{C_i}$. 
		\end{itemize}
		
		\item If $C$ is constructed by a series composition, then for any $i\in\{1,2\}$ we consider again two cases:
		\begin{itemize}
			\item $\load^{*}_C=1$: it holds by assumtion that $\psi_C = \phi_C(1)$ and by Remark~\ref{rem:minLoadAndWeightInCs} that $\load^{*}_{C_i}=1$. Therefore, by definition, $\psi_{C_i}= \psi_C\frac{\phi_{C_i}(1)}{\phi_{C}(1)}=\phi_{C_i}(1)$, which is what is required for $\load^{*}_{C_i}=1\,.$
			
			\item $\load^{*}_C>1$: here it is either $\psi_{C_i} =\phi_{C_i}(1) > \frac{\phi_{C_i}(\load)}{\load}$ for any $\load>1$, or $\psi_{C_i} \geq \psi_C\frac{\phi_{C_i}(\load^{*}_C)}{\phi_{C}(\load^{*}_C)}.$ In the second case:
			$$\psi_{C_i} \geq \psi_C\frac{\phi_{C_i}(\load^{*}_C)}{\phi_{C}(\load^{*}_C)}>\frac{\phi_C(\load^{*}_C)}{\load^{*}_C}\frac{\phi_{C_i}(\load^{*}_C)}{\phi_{C}(\load^{*}_C)}=\frac{\phi_{C_i}(\load^{*}_C)}{\load^{*}_C}\geq \frac{\phi_{C_i}(\load)}{\load}\,.$$
		\end{itemize}
	\end{itemize}
\end{proof}

\subsection{Proof of Lemma~\ref{lem:costShares}}
\begin{proof}
	We distinguish between two cases:
	\begin{itemize} 
		\item $e\notin \opt(\load_e(\s))$: for $i=h_e(\s)$, $\xi_{ie}(\s)=c_e(\load_e(\s))\geq c_e(1)=\phi_e(1)\geq \psi_e$ by Lemma~\ref{lem:psi<=c1}. The rest of the players are charged $0$.
		
		\item $e\in \opt(\load_e(\s))$: for $i= h_e(\s)$ then $\xi_{ie}(\s)=\psi_e$ by definition. If it exists $i\neq h_e(s^*)$ then $\load_e(\s) >1$. Additionally, since $e\in \opt(\load_e(\s))$, $\load_e(\s) \geq \load^{*}_e$. By using Lemma~\ref{lem:psi} twice, 
		$$\xi_{ie}(s^*) = \frac{c_e(\load_e(\s)) - \psi_e}{\load_e(\s) - 1}<\frac{c_e(\load_e(\s)) - \frac{c_e(\load_e(\s))}{\load_e(\s)}}{\load_e(\s) - 1} = \frac{c_e(\load_e(\s))}{\load_e(\s)}<\psi_e\,.$$
	\end{itemize}
\end{proof}

\subsection{Proof of Theorem~\ref{thm:multicastLB1}}
\begin{proof}
  This follows directly from the instance of Proposition 4.12 of
  \cite{CRV10} and we include it only for completeness.  Consider an
  instance with $n$ players, each one with a source $s_i$ and a common
  sink $t$ as in Figure~\ref{fig:multiBB}, and the constant edge costs
	are as they appear in the figure with the cost of the $(v, t)$ edge
	being equal to $1$. It is easy to verify that the strategy profile
	where every player $i$ uses the direct edge $(s_i, t)$ to get to the
	sink is an equilibrium for any budget-balanced resource-aware protocol.
	In particular, since there is only a single player that uses each $(s_i, t)$
	edge, every budget-balanced protocol will need to charge each of these agents
	a cost of 1. Any unilateral deviation would also cost them $1$, however,
	due to the cost of $(v, t)$. In the optimal solution all agents share 
	the $(v, t)$ edge and the social cost is 1 instead of $n$.
\end{proof}

\subsection{Proof of Lemma~\ref{lem:NoPlayersInOpt}}
\begin{proof}
Starting from any optimum $\s^*$ we swap subpaths of carefully selected players and end up to some optimum satisfying the lemma's statement. 
The swaps appear in rounds in each component starting from $G$ and following its decomposition $(C_r, \ldots, C_1)$ (recall that $G=C_r$). We say that a strategy $\str_i$ passes through component $C$ if it uses some of its edges. 

  For each component $C$ (starting from $C_r$ and reducing the index) we reallocate the optimum inside $C$ by using the following reallocation procedure:
\begin{itemize}
\item If $C$ is constructed by a series composition, we do not do any reallocation and proceed to the next component.
\item If $C$ is constructed by a parallel composition of $C_1$ and $C_2$, for players $i$ from $1$ to $n$, if 
\begin{itemize}
\item both $\str_i$ and $\str^*_i$ pass through $C$, 
\item w.l.o.g. $\str_i$ passes through $C_1$ and $\str^*_i$ passes through $C_2$ and 
\item there exists $j>i$ such that $\str^*_j$ passes through $C_1$
\end{itemize}
switch the subpaths of $\str^*_i$ and $\str^*_j$ inside $C$ 
\end{itemize}

Note that after the above procedure, the new $\s^*$ has the same total cost with the initial profile and therefore is still an optimum. For the final optimum $\s^*$ (after considering all components) we prove the following claim.

\begin{claim}
	\label{cl:convexComp}
For any player $j$ and any component $C$ that $\str^*_j$ passes through, it holds that for any player $i<j$, if $\str_i$ passes through $C$ so does $\str^*_i$.
\end{claim}
\begin{proof}

We prove the claim by induction starting from $G$ and following its decomposition.  
The claim is trivially satisfied for $G$. 

 Suppose that component $C$ satisfies the claim; we will show that also $C_1$ and $C_2$ that compose $C$ satisfy the claim. We distinguish between the two possible cases of how $C$ was composed.
 
 {\bf $\mathbf{C}$ is constructed by a series composition of $\mathbf{C_1, C_2}$}. It is easy to see that any path passes through $C$ should also passes through both $C_1$ and $C_2$ and therefore both $C_1,C_2$ satisfy the statement of the claim.
  
  {\bf $\mathbf{C}$ is constructed by a parallel composition of $\mathbf{C_1, C_2}$}. For the sake of contradiction assume that $C_1$ doesn't satisfy the statement and so there exists $j$ and $i$ with $i<j$ and such that $\str_i$ and $\str^*_j$ pass through $C_1$ but not $\str^*_i$; let $i$ be the minimum such index. Since $C$ satisfies the claim's statement, $\str^*_i$ should pass through $C$ and therefore through $C_2$. Note though that while processing $i$ during the round considering $C$ in the above reallocation procedure, there shouldn't exist any $\str^*_{j'}$ with $i<j'$ passes through $C_1$ otherwise $\str^*_i$ should have been swapped with that; therefore there shouldn't exist $\str^*_j$. However, since we assumed that $\str^*_j$ passes through $C_1$ ultimately, it should have been swapped afterwards with some other optimal subpath, but this proves the existence of some $\str^*_{j'}$ passing through $C_1$ with $i<j'<j$ while processing $i$.\footnote{Note that the order we process $C$'s in the reallocation procedure guarantees that any reallocation after processing $C$ does not affect the input of $C$, so $\str^*_j$ cannot be swapped to pass through $C_1$ after processing $C$.} This is a contradiction meaning that $C_1$ should satisfy the statement as well.
\end{proof}

Claim \ref{cl:convexComp} holds for every edge $e\in \str^*_i$ and trivially we get $\load_e^{< i}(\s) \leq \load_e^{< i}(\s^*)$. 
\end{proof}

\subsection{Proof of Theorem~\ref{poadagconvexovercharging}}
\begin{proof}
	We use the DAG shown in Figure \ref{fig:overchargingexample} to prove the statement. In this graph all the players want to go from $s$ to $t$ and the cost of each edge is as indicated in the figure if a single player uses the edge and infinity if more than one players use it.
	
	\begin{figure}[h]
		\centering
		\tikzset{middlearrow/.style={
				decoration={markings,
					mark= at position 0.5 with {\arrow{#1}} ,
				},
				postaction={decorate}
			}
		}
		\begin{tikzpicture}[xscale=1,yscale=1]
		\node[draw,circle,thick,minimum width =0.7 cm] at (0,0) (s) {$\mathbf{s}$} ;
		\node[draw,circle,thick,minimum width = 0.7 cm] at (4,1.5) (v) {$\mathbf{v}$} ;
		
		\node[draw,circle,thick,minimum width =0.7 cm] at (4,-1.5) (u) {$\mathbf{u}$} ;
		
		\node[draw,circle,thick,minimum width =0.7 cm] at (8,0) (t) {$\mathbf{t}$} ;

		\draw[thick, middlearrow={stealth}] (s) to node [above, sloped] {{$\frac{\sqrt{33}-1}{8}, \infty$}} (v);
		\draw[thick, middlearrow={stealth}] (s) to node [below, sloped] {{$0, \infty$}} (u);
		\draw[thick, middlearrow={stealth}] (v) to node [above, sloped] {{$0, \infty$}}(t);
		\draw[thick, middlearrow={stealth}] (u) to node [below, sloped] {{$\frac{\sqrt{33}-1}{8}, \infty$}} (t);
		\draw[thick, middlearrow={stealth}] (u) to node [above, sloped] {{$0, \infty$}}(v);
		
		\draw[thick, middlearrow={stealth}] (s) to [out=85, in=95] node [above] {{$1, \infty$}} (t);
		
		\end{tikzpicture}
		\caption{A simple acyclic graph with a single source $s$ and sink $t$. All of the cost functions are infinity for load $\load\geq 2$.  This graph is used to prove Theorem~\ref{poadagconvexovercharging}.}
		\label{fig:overchargingexample}
	\end{figure}
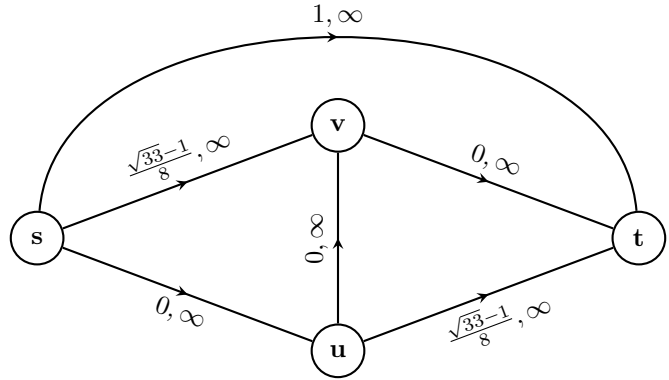

	Considering three players  $a, b$ and $c$, in the optimum one player uses the upper edge, and the other two the $\mathbf{svt}$ and the $\mathbf{sut}$ paths. The total cost of the optimum is $\frac{\sqrt{33}-1}{4} + 1$. If any other strategy profile was an equilibrium it would result in an unbounded total cost and give an unbounded PoA. Therefore, there should exist an optimum strategy profile $\s^*$ that it is a Nash equilibrium. Assume w.l.o.g. that in $\s^*$, $a$ uses the upper edge, $b$ uses the $\mathbf{svt}$ path, and $c$ uses the $\mathbf{sut}$ path. 
	
	If we consider now the instance in which only $b$ and $c$ use the network, either the strategy profile in which $b$ uses the $\mathbf{svt}$ path, and $c$ uses the $\mathbf{sut}$ path is an equilibrium, or one of $b$ or $c$ pays at least $1$ in $\s^*$ (due to some overcharging). The reason is that since $\s^*$ is a Nash equilibrium, no player has an incentive to deviate to any other path but the upper one, and if neither $b$, nor $c$ pays more than $1$, obviously none of them would prefer the upper edge. 
	
	For the first case, the cost of this strategy profile is $\frac{\sqrt{33}-1}{4}$, and the optimum with two players is one player to use the zig-zag path $\mathbf{suvt}$, and the other to use the upper edge with total cost $1$. This results in PoA $\geq \frac{\sqrt{33}-1}{4}$. 
	
	For the second case, if one of $b$ or $c$ pays at least 1 in $\s^*$, the cost of $\s^*$ becomes at least $\frac{\sqrt{33}-1}{8}+2$ and we have
	$$\text{PoA} \geq \frac{\frac{\sqrt{33}-1}{8}+2}{\frac{\sqrt{33}-1}{4}+1} = \frac{\sqrt{33}-1}{4}\,.$$
	
\end{proof}

\subsection{Proof of Theorem~\ref{multidagconvex}}
\begin{proof}
	We use the graph shown in Figure \ref{fig:multicastexample} to prove the statement. 
	
	\begin{figure}[h]
	\centering
	\tikzset{middlearrow/.style={
			decoration={markings,
				mark= at position 0.5 with {\arrow{#1}} ,
			},
			postaction={decorate}
		},
		dots/.style={decoration={markings, mark=between positions 0 and 1 step 10pt with { \draw [fill] (0,0) circle [radius=1.0pt];}},
			postaction={decorate}
		}
	}
	\begin{tikzpicture}[xscale=0.9,yscale=0.8]
	\node[draw,circle,thick,minimum width =0.9 cm] at (-1,8) (s1) {$\mathbf{s_1}$} ;
	\node[draw,circle,thick,minimum width =0.9 cm] at (1,8) (s2) {$\mathbf{s_2}$} ;
	\node[draw,circle,thick,minimum width =0.9 cm] at (6.5,5) (v1) {$\mathbf{v_n}$} ;
	\node[draw,circle,thick,minimum width =0.9 cm] at (4.84,5) (v2) {$\mathbf{u_n}$} ;
	\node[draw,circle,thick,minimum width =0.9 cm] at (3.16,5) (v3) {$\mathbf{v_{n-1}}$} ;
	\node[draw,circle,thick,minimum width =0.9 cm] at (1.5,5) (vx) {$\mathbf{u_{n-1}}$} ;
	\node[draw,circle,thick,minimum width =0.9 cm] at (-1.5,5) (v4) {$\mathbf{u_{2}}$} ;
	\node[draw,circle,thick,minimum width =0.9 cm] at (-3.5,5) (v5) {$\mathbf{v_{1}}$} ;
	\node[draw,circle,thick,minimum width =0.9 cm] at (-5.5, 5) (v6) {$\mathbf{u_{1}}$} ;
	\node[draw,circle,thick,minimum width =0.9 cm] at (0, 2) (t) {$\mathbf{t}$} ;
	
	\path[dots] (0.45,5) to (-0.55,5);
	
	\draw[thick, middlearrow={stealth}] (s2) to node [above, sloped] {{$0, \infty$}} (v1);
	\draw[thick, middlearrow={stealth}] (v2) to node [above, sloped] {{$0, \infty$}} (t);
	\draw[thick, middlearrow={stealth}] (s1) to node [above, sloped] {{$0, \infty$}} (v3);
	\draw[thick, middlearrow={stealth}] (v4) to node [above, sloped] {{$0, \infty$}} (t);
	\draw[thick, middlearrow={stealth}] (vx) to node [above, sloped] {{$0, \infty$}} (t);
	\draw[thick, middlearrow={stealth}] (s1) to node [above, sloped] {{$0, \infty$}} (v5);
	\draw[thick, middlearrow={stealth}] (v6) to node [below, sloped] {{$0, \infty$}} (t);
	
	\draw[thick, middlearrow={stealth}] (v1) to node [below, sloped] {{$1, \infty$}} (v2);
	\draw[thick, middlearrow={stealth}] (v3) to node [above, sloped] {{$0, \infty$}} (v2);
	\draw[thick, middlearrow={stealth}] (v3) to node [below, sloped] {{$1, \infty$}} (vx);
	\draw[thick, middlearrow={stealth}] (v5) to node [above, sloped] {{$1, \infty$}} (v6);	
	\draw[thick, middlearrow={stealth}] (v5) to node [below, sloped] {{$0, \infty$}} (v4);	
	
	\draw[thick, middlearrow={stealth}] (s2.east) to [out=0,in=90]  ([xshift={0.5cm}]v1.east) node[right]{{$\sqrt{n}, \infty$}} to [out=-90, in=0] (t.east);

	\end{tikzpicture}
	\caption{An acyclic graph with a two sources $s_1$ and $s_2$, and sink $t$. All of the cost functions are infinity for load $\load\geq 2$. This graph is used to prove Theorem~\ref{multidagconvex}.}
	\label{fig:multicastexample}
\end{figure}
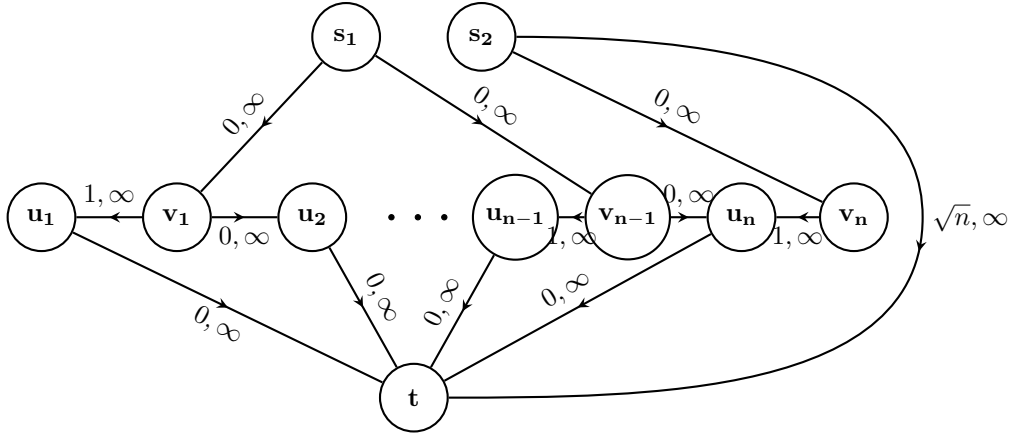

	We consider $n+1$ players, $\{1,2,\ldots, n, n+1\}$, where the first $n-1$ players want to connect $s_1$ to $t$, and players $n, n+1$ want to connect $s_2$ to $t$. In the optimum, the paths $\mathbf{s_1v_iu_it}$ for $1\leq i <n$, $\mathbf{s_2v_nu_nt}$ and $\mathbf{s_2t}$ are used. The total cost of the optimum is $n+\sqrt{n}$. If any other strategy profile was an equilibrium it would result in an unbounded total cost and give an unbounded PoA. Therefore, there should exist an optimum strategy profile $\s^*$ that is a Nash equilibrium. Assume w.l.o.g. that in $\s^*$, each player $i < n$ uses the $\mathbf{s_1v_iu_it}$ path, player $n$ uses the $\mathbf{s_2v_nu_nt}$ path  and player $n+1$ uses the edge $\mathbf{s_2t}$.  	
	Now we consider two cases:
	
	\textbf{Case 1:} Player $n$ is charged at most $\sqrt{n}$ for the $\mathbf{s_2v_nu_nt}$ path.
	
	In this case we consider the instance in which only the first $n$ players use the network. The strategy profile in which each player $i < n$ uses the $\mathbf{s_1v_iu_it}$ path and player $n$ uses the $\mathbf{s_2v_nu_nt}$ path should be an equilibrium. The reason is that since $\s^*$ is a Nash equilibrium, no player has an incentive to deviate to any other path but the $\mathbf{s_2t}$, and the only player that might have an incentive to do this is player $n$. Since player $n$ pays less than $\sqrt{n}$ in $\mathbf{s_2v_{n}u_{n}t}$, she does not have an incentive to deviate, and hence this strategy profile is an equilibrium with the total cost of at least $n$. However, the optimal strategy profile for these players is for each player $i < n$ to use the $\mathbf{s_1v_iu_{i+1}t}$ path (each with zero cost) and player $n$ to use edge $\mathbf{s_2t}$ with total cost of $\sqrt{n}$. Therefore, PoA $\geq \sqrt{n}$ for this case.
	
	\textbf{Case 2:} Player $n$ is charged more than $\sqrt{n}$ for the $\mathbf{s_2v_nu_nt}$ path.
	
	In this case we consider a single player who wants to connect $s_2$ to $t$. The optimal cost is $1$ (considering the original costs), but with the assumption that the player should pay at least $\sqrt{n}$, PoA $\geq \sqrt{n}$ in this case too.
\end{proof}

\end{document}